\documentclass[11pt,letterpaper]{article}

\usepackage[affil-it]{authblk}

\usepackage{graphicx, verbatim}
\usepackage{amsmath,amsfonts,amsthm}
\usepackage{braket}
\usepackage{fullpage}
\usepackage[normalem]{ulem}

\usepackage[colorlinks=true,linkcolor=blue,citecolor=blue,urlcolor=blue]{hyperref}

\usepackage[dvipsnames]{xcolor}

\usepackage{pxfonts}

\def\I{\mbox{$1 \hspace{-1.0mm}  {\bf l}$}}
\def\id{\mathrm{id}}
\DeclareMathOperator{\Tr}{Tr}
\DeclareMathOperator{\supp}{supp}
\DeclareMathOperator{\diag}{diag}

\def\CC{\mathbb{C}}
\def\RR{\mathbb{R}}

\newcommand{\mc}{\mathcal}

\newcommand{\mbb}{\mathbb}
\newcommand{\msf}{\mathsf}

\newcommand{\tr}{\text{Tr}}

\def\calD{\mathcal{D}}

\def\calH{\mathcal{H}}
\def\calL{\mathcal{L}}

\def\calS{\mathcal{S}}

\newcommand{\ketbra}[2]{\ket{#1}\!\bra{#2}}

\let\emphb\relax 
\DeclareTextFontCommand{\emphb}{\bfseries\em}

\newtheorem{theorem}{Theorem}
\newtheorem{corollary}[theorem]{Corollary}
\newtheorem{lemma}[theorem]{Lemma}

\newtheorem{proposition}[theorem]{Proposition}

\theoremstyle{definition}
\newtheorem{definition}{Definition}[section]

\begin{document}
\title{Entanglement manipulation and distillability beyond LOCC}

\author[a]{Eric Chitambar}
\author[b]{Julio I.\ de Vicente}
\author[c,d]{Mark W.\ Girard}
\author[c]{Gilad Gour}
\affil[a]{Department of Electrical and Computer Engineering, Coordinated Science Laboratory\\ University of Illinois at Urbana-Champaign, Urbana, IL 61801, USA}
\affil[b]{Departamento de Matem\'aticas, Universidad Carlos III de Madrid, E-28911  Legan\'es (Madrid), Spain}
\affil[c]{Department of Mathematics and Statistics, University of Calgary, Alberta T2N 1N4, Canada}
\affil[d]{Institute for Quantum Computing, University of Waterloo, Ontario N2L 3G1, Canada}


\maketitle
\begin{abstract}


When a quantum system is distributed to spatially separated parties, it is natural to consider how the system evolves when the parties perform local quantum operations with classical communication (LOCC).  However, the structure of LOCC channels is exceedingly complex leaving many important physical problems unsolved. In this paper we consider generalized resource theories of entanglement based on different relaxations to the class of LOCC. The behavior of various entanglement measures is studied under non-entangling channels, as well as the newly introduced classes of dually non-entangling and PPT-preserving channels.  In an effort to better understand the nature of LOCC bound entanglement, we study the problem of entanglement distillation in these generalized resource theories. We first show that unlike LOCC, general non-entangling maps can be superactivated, in the sense that two copies of the same non-entangling map can nevertheless be entangling.  On the single-copy level, we demonstrate that every NPT entangled state can be converted into an LOCC-distillable state using channels that are both dually non-entangling and having a PPT Choi representation and that every state can be converted into an LOCC-distillable state using operations belonging to any family of polytopes that approximate LOCC. We then turn to the stochastic convertibility of multipartite pure states and show that any two states can be interconverted by any polytope approximation to the set of separable channels. Finally, as an analog to $k$-positive maps, we introduce and analyze the set of $k$-non-entangling channels.


\end{abstract}
\section{Introduction}
\label{intro}

In the distant laboratory paradigm of quantum information theory, a multipartite quantum system is distributed to spatially separated parties.  The individual parties are restricted to performing only local quantum operations on their respective subsystems, but they are permitted to coordinate their actions through global classical communication.  Quantum channels arising from this scenario are known as LOCC (local operations and classical communication), and they emerge as the natural operational class to consider for distributed quantum information processing tasks.  Entanglement becomes a resource under the restriction of LOCC; it is consumed when performing information processing tasks (such as quantum teleportation), and entangled states cannot be generated using LOCC alone. It is therefore natural to study entanglement in the framework of a \emph{quantum resource theory}, which characterizes the convertibility of entangled states under LOCC and investigates their applicability in various information-theoretic tasks (see the reviews \cite{Horodecki2009, Plenio2007, Chitambar-2019a}).

Despite the conceptually intuitive description of LOCC in the distant-lab setting, mathematically characterizing the set of LOCC operations is notoriously complex \cite{Chitambar2014}.  One way to overcome this difficulty is by relaxing the constraints of LOCC to encompass a more mathematically simple class of operations.  Any task found to be impossible by these more general operations is therefore also impossible by LOCC.  In fact, LOCC is not the largest class of channels that cannot create entanglement from unentangled, or separable, states.  There exists a plethora of operational classes ``beyond'' LOCC that act invariantly on the set of separable states, and each of these defines a different quantum resource theory of entanglement manipulation.

It is interesting to study what differences arise in these generalized resource theories of entanglement compared to the standard LOCC scenario.  Doing so has already proven successful in the asymptotic scenario where, under LOCC, certain mixed states demonstrate irreversibility in terms of pure-state entanglement distillation and the reverse process of entanglement dilution \cite{Vidal2001,Yang2005}.  The most dramatic example of this irreversibility is the phenomenon of \emph{bound entanglement} \cite{Horodecki1998a}, which refers to entangled mixed states that require the consumption of pure-state entanglement for their preparation, and yet no pure-state entanglement can be distilled back from them.  Remarkably, the theory of entanglement distillation and dilution becomes much more amenable and elegant after enlarging the class of operations to include all asymptotically non-entangling operations.   The work in \cite{Brandao2008,Brandao2010,Brandao2011} shows that bound entanglement no longer exists under this larger class of operations, and more importantly, the entanglement in any mixed state can be reversibly distilled and diluted at a rate given by the regularized relative entropy of entanglement.

In this paper we analyze in detail entanglement manipulation using operations beyond LOCC.  Specifically, we study which state transformations are possible and which entanglement measures retain their validity within generalized resource theories of entanglement.  For example, \emph{separable} maps are those whose Choi representations are separable \cite{Rains1997}, while \emph{non-entangling} maps are precisely those that do not produce an entangled output state for any separable input.  The class of \emph{dually non-entangling} maps, which we introduce here, consists of channels $\Lambda$ such that both $\Lambda$ and its dual map $\Lambda^*$ are non-entangling.  All classes of operations that we consider contain LOCC as a proper subset.

One of our motivations is to better understand the nature of bound entanglement within the standard framework of LOCC.  While it is known that no pure-state entanglement can be distilled from any state having a positive partial transpose (PPT), one of the major open problems in quantum information theory is to determine if the converse is also true \cite{pankowski2010}; namely, does there exist non-PPT (NPT) bound entanglement?  A strategy that would allow us to answer this question in the affirmative would be to identify a class of operations that contains LOCC and which lacks the ability to distill certain NPT states.

The first comprehensive attempt in this direction was carried out in \cite{Eggeling2001}, where it was shown that NPT bound entanglement does not exist when the class of operations is enlarged to contain all so-called PPT operations, as originally introduced in \cite{Rains1999,Rains2000}. This is the class of channels whose Choi representation is PPT (see also Sect.\ \ref{sec:prelims} for definitions). Even stronger, the results of \cite{Brandao2008,Brandao2010,Brandao2011} imply that no bound entanglement (not even PPT bound entanglement) exists for the class of all strictly non-entangling channels.

Beyond the question of distillability, we also examine the difference between the class of channels whose Choi operator is PPT and the class of PPT-preserving channels.  While the former refers to the well-known class of PPT channels originally defined by Rains \cite{Rains1999,Rains2000}, to our knowledge the latter has not been previously studied. Finally, analogous to the relation between completely positive and $k$-positive maps, we also introduce here the classes of \emph{$k$-non-entangling} and \emph{completely non-entangling} maps. The former refers to maps that are non-entangling when tensored with the identity on a $k$-dimensional ancilla, while the later refers to maps that are $k$-non-entangling for all $k$.

The rest of the paper is structured as follows. Section \ref{sec:prelims} introduces the necessary background and notation, including precise definitions for the classes of channels that we study. Section \ref{sec:nonenttransformations} investigates how certain entanglement measures behave under different classes of non-entangling channels.  The entanglement measures studied in this section include the robustness of entanglement, Schmidt rank, the R\'enyi $\alpha$-entropies of entanglement (as well as the relative R\'enyi entropies of entanglement), and the negativity.  Although the robustness remains monotonic under non-entangling channels, we use the robustness to derive conditions for the convertibility of pure states under non-entangling channels that are independent of the celebrated majorization criterion for LOCC channels \cite{Nielsen1999}.  We also show that the R\'enyi $\alpha$-entropies of entanglement (which are entanglement measures for all $\alpha\in[0,+\infty]$) can be increased arbitrarily under non-entangling channels for $\alpha\in[0,1/2)$.  On the other hand, we show that the $\alpha$-entropies of entanglement in the range $\alpha\in[1/2,+\infty]$ coincide with a related relative R\'enyi entropy of entanglement measure, and therefore they cannot increase under non-entangling channels. 
In particular, this gives a closed-form expression for the relative $\alpha$-entropies of entanglement in certain cases. We also demonstrate that the Schmidt rank can be increased under the more restrictive class of dually non-entangling maps. Finally, we show that the negativity can be increased by an arbitrarily large fractional amount under PPT-preserving maps, thus highlighting the difference between the classes of PPT and PPT-preserving maps.

Entanglement distillation by classes of channels larger than LOCC is investigated in Section \ref{sec:distillbeyondLOCC}. We demonstrate that all entangled states can be converted to an LOCC-distillable state using dually non-entangling channels. The significance of this result is somewhat tempered by the fact that two copies of a non-entangling channel need not be non-entangling, something we also demonstrate in Section \ref{sec:distillbeyondLOCC}.  Consequently, single-copy convertibility of a given state to an LOCC-distillable one by some generalized class of operations does not imply that the original state is distillable within the generalized resource theory. Finally, we close the section by showing that any class of operations that generalizes LOCC by placing a finite number of linear constraints on the Choi matrix always allows to obtain LOCC-distillable states from any quantum state.



In Section \ref{sec:stochastic}, we relax the constraint of trace-preserving maps and turn to the problem of stochastic convertibility between pure states.  For multipartite state spaces, a natural way to categorize entanglement is in terms of stochastic convertibility under LOCC.  That is, two states $\ket{\psi}$ and $\ket{\hat{\psi}}$ belong to the same entanglement class if there is invertible LOCC transformation from $\ket{\psi}$ to $\ket{\hat{\psi}}$ that succeeds with some nonzero probability \cite{Dur2000}. We show that this entanglement structure completely collapses under any class of entangling undetected channels.  Connections to the problem of tensor rank calculation are also discussed.

The structure of $k$-non-entangling maps and $k$-PPT-preserving maps are analyzed in Section \ref{sec:kresourcepreserving}. We show that complete non-entangling is equivalent to $d$-non-entangling  when the local systems have dimension $d$, and we show that the structure of $k$-non-entangling maps is nontrivial.  That is, for every $k<d$ there exists $k$-non-entangling maps that are not $(k+1)$-non-entangling.  We also investigate distillability under this class of maps but we have not been able to answer whether or not they allow for bound entanglement. In particular, we show that previously used constructions are not 3-non-entangling, so the weaker question of whether every entangled state can be converted to an LOCC-distillable one under this class of maps remains open. Lastly, a detailed proof of Theorem \ref{thm:EalphaisERalphainverse} is presented in Section \ref{sec:Salphaproof} while a concluding discussion is given in Section \ref{sec:discussion}.

\section{Notation and definitions}
\label{sec:prelims}


Let $\calL(\calH)$ denote the space of linear operators on a finite-dimensional Hilbert space $\calH$ with orthonormal basis denoted by $\{\ket{1},\dots,\ket{n}\}$ (where $n$ is the dimension of the space). An operator $\rho\in\calL(\calH)$ is in the set $\calD(\calH)$ of density operators (or states) on $\calH$ if $\rho\geq0$ and $\Tr(\rho)=1$. The tensor product of two Hilbert spaces $\calH_{\msf{A}}$ and $\calH_{\msf{B}}$ is denoted by $\calH_{\msf{A}}\otimes\calH_{\msf{B}}$ or $\calH_{\msf{A}\msf{B}}$. The identity operator on a Hilbert space $\calH_{\msf{A}}$ is denoted $\I_{\msf{A}}$, while the identity map on $\calL(\calH_{\msf{A}})$ is denoted $\id_{\msf{A}}$ and the identity operator on $\CC^d\otimes\CC^d$ is denoted by $\I_{d^2}$.

We write $\psi$ to denote the density operator $\psi=\ketbra{\psi}{\psi}$ associated with a pure state vector $\ket{\psi}\in\calH$. For bipartite pure states $\ket{\psi}\in\calH_{\msf{A}\msf{B}}$ we typically assume, without loss of generality, that $\ket{\psi}$ is in Schmidt form,
\[
 \ket{\psi} = \sum_{i=1}^d\sqrt{\lambda_i}\ket{ii},
\]
where $\ket{ii}$ is short-hand for $\ket{i}\otimes\ket{i}$, and $\lambda=(\lambda_1,\dots,\lambda_d)$ are the Schmidt coefficients of $\psi$ (in decreasing order) such that $\sum_{i=1}^d\lambda_i=1$ and $d\leq\min\{\dim\calH_\msf{A},\dim\calH_\msf{B}\}$. The Schmidt rank of the state is defined to be the number of non-zero entries in $\lambda$, which we denote by $\operatorname{rank}(\lambda)$.  A positive operator $\sigma\in\calL(\calH_{\msf{A}\msf{B}})$ is said to be separable with respect to $\msf{A}\!:\!\msf{B}$ if it can be written as $\sigma=\sum_{i} \tau_i\otimes\omega_i$ for some positive operators $\tau_i\in\calL(\calH_{\msf{A}})$ and $\omega_i\in\calL(\calH_{\msf{B}})$. The set of separable density operators on $\calH_{\msf{A}\msf{B}}$ will be denoted $\calS(\calH_{\msf{A}:\msf{B}})$, or simply $\calS$ if the spaces $\msf{A}$ and $\msf{B}$ are clear from context.

The partial transpose of an operator $X\in\calL(\calH_{\msf{A}\msf{B}})$, denoted $X^{\Gamma_{\msf{A}}}$, is the operator that results from applying the transposition map $A\mapsto A^T$ (with respect to a given basis of $\calH_{\msf{A}}$) to system~$\msf{A}$. We simply write $X^{\Gamma}$ if the system being transposed is clear from context. A positive operator $\rho\in\calL(\calH_{\msf{A}\msf{B}})$ is positive under partial transpose (PPT) with respect to \mbox{$\msf{A}\!:\!\msf{B}$} if $\rho^\Gamma\geq0$, and it is otherwise said to be non-PPT (or NPT). (We note that the definition of PPT is independent of the basis chosen for the respective transposition.)  It is well-known that every separable state is PPT, with the converse holding true only in dimensions $2\otimes 2$ and $2\otimes 3$ \cite{Horodecki1996}.

For any integer $d\geq2$, the maximally entangled pure state on $\CC^d\otimes\CC^d$ is denoted
\[
 \ket{\phi^+_d} = \frac{1}{\sqrt{d}}\sum_{i=1}^d\ket{ii},
\]
and we write $\phi^+_d=\ketbra{\phi^+_d}{\phi^+_d}$ to denote the corresponding density operator. We denote the so-called `flip' operator on $\CC^d\otimes\CC^d$ by
\[
 F_d = \sum_{i,j=1}^d\ketbra{ij}{ji},
\]
which is precisely the partial transpose of $\phi^+_d$ scaled by $d$. Analogously, we write $\phi^+_{\msf{A}}$ to denote the (unnormalized) maximally entangled operator on the Hilbert space $\calH_{\msf{A}}\otimes\calH_{\msf{A}}$,
\[
 \phi^+_{\msf{A}} = \sum_{i,j=1}^{\dim(\calH_{\msf{A}})}\ketbra{i}{j}\otimes\ketbra{i}{j},
\]
where $\{\ket{i}\}$ is an orthonormal basis for $\calH_{\msf{A}}$. Note that we can write the maximally entangled operator of the four-partite space $\phi^+_{\msf{A}\msf{B}}\in\calL(\calH_{\msf{A}\msf{B}}\otimes\calH_{\msf{A}\msf{B}})$ as
\[
 \phi^+_{\msf{A}\msf{B}} = \phi^+_{\msf{A}}\otimes\phi^+_{\msf{B}}.
\]

Given a linear map $\Lambda:\calL(\calH_{\msf{A}_1})\rightarrow\calL(\calH_{\msf{A}_2})$, its Choi representation is the linear operator $J(\Lambda)\in\calL(\calH_{\msf{A}_2\msf{A}_1})$ defined by
\[
 J(\Lambda) = \Lambda\otimes\id_{\msf{A}_1} (\phi^+_{\msf{A}_1}).
\]
For any $X\in\calL(\calH_{\msf{A}_1})$, the action of the map $\Lambda$ on $X$ can be written in terms of its Choi representation as $\Lambda(X)= \Tr_{\msf{A}_1}(J(\Lambda)(\I_{\msf{A}_2}\otimes X^T))$. Recall that $\Lambda$ is completely positive (CP) if and only if its Choi representation is a positive semidefinite operator, and $\Lambda$ is trace preserving (TP) if and only if it holds that $\Tr_{\msf{A}_2}(J(\Lambda)) = \I_{\msf{A}_1}$. A linear map that is both completely positive and trace preserving (CPTP) is called a \emph{quantum channel}. Given two density operators $\rho_1,\rho_2\in\calD(\calH')$ and an operator $A\in\calL(\calH)$, the linear map $\Lambda:\calL(\calH)\rightarrow\calL(\calH')$ defined by
\begin{equation}\label{eq:LambdaA}
 \Lambda (X) = \Tr(AX)\rho_1 + \Tr((\I-A)X)\rho_2
\end{equation}
is CPTP as long as $0\leq A\leq \I$. Many of the maps constructed in this work take this form. Given a linear map $\Lambda:\calL(\calH)\rightarrow\calL(\calH')$, its dual map $\Lambda^*:\calL(\calH')\rightarrow\calL(\calH)$ is the unique linear map such that $\Tr(\Lambda(X)Y)=\Tr(X\Lambda^*(Y))$ holds for all operators $X\in\calL(\calH)$ and $Y\in\calL(\calH')$. It follows readily from the definition that the dual of a map of the form in \eqref{eq:LambdaA} is given by
\[
 \Lambda^*(Y) = \Tr(\rho_1Y)A + \Tr(\rho_2Y)(\I-A),
\]
while the Choi representation of a map of the form in \eqref{eq:LambdaA} is given by
\[
 J(\Lambda) = \rho_1\otimes A^T + \rho_2\otimes(\I-A)^T.
\]

\subsection{Definitions of classes of channels}

We now define all of the classes of operations that will be explored in this paper. We first recall the definitions of separable and PPT maps. Let $\Lambda:\calL(\calH_{\msf{A}_1\msf{B}_1})\rightarrow\calL(\calH_{\msf{A}_2\msf{B}_2})$ be a completely positive linear map. Then $\Lambda$ is said to be separable if its Choi representation $J(\Lambda)\in\calL(\calH_{\msf{A}_2\msf{B}_2}\otimes\calH_{\msf{A}_1\msf{B}_1})$ is \emph{separable} with respect to \mbox{$\msf{A}_2\msf{A}_1\!:\!\msf{B}_2\msf{B}_1$}.  To keep the bipartition clear, we will often denote $\calH_{\msf{A}_2\msf{B}_2}\otimes\calH_{\msf{A}_1\msf{B}_1}$ as $\calH_{\msf{A}_2\msf{A}_1:\msf{B}_2\sf{B}_1}$.  Note that $\Lambda$ is separable if and only if it can be written as
\[
 \Lambda = \sum_{j} \Phi_j \otimes \Psi_j
\]
for some CP maps $\Phi_j:\calL(\calH_{\msf{A}_1})\rightarrow\calL(\calH_{\msf{A}_2})$ and $\Psi_j:\calL(\calH_{\msf{B}_1})\rightarrow\calL(\calH_{\msf{B}_2})$. Analogously, the map $\Lambda$ is a \emph{PPT map} if the partial transpose of its Choi representation is positive semi-definite, i.e., if $J(\Lambda)^\Gamma=J(\Lambda)^{\Gamma_{\msf{A}_2\msf{A}_1}}\geq0$. It holds that $\Lambda$ is PPT if and only if the map $\Lambda^\Gamma$ defined by
\[
 \Lambda^\Gamma(X) = (\Lambda(X^{\Gamma_{\msf{A}_1}}))^{\Gamma_{\msf{A}_2}}
\]
is also completely positive. In particular, note that $J(\Lambda^\Gamma) = J(\Lambda)^\Gamma$. This definition of PPT maps is due to Rains \cite{Rains1999}.

In this paper, we are primarily concerned with the \emph{non-entangling} maps, which are defined explicitly as follows. We provide here an analogous definition for PPT-preserving maps.

\begin{definition}
Let $\Lambda:\calL(\calH_{\msf{A}_1\msf{B}_1})\rightarrow\calL(\calH_{\msf{A}_2\msf{B}_2})$ be a completely positive map. Then $\Lambda$ is \emphb{separable-preserving} (or \emphb{non-entangling}) if $\Lambda(\sigma)$ is separable (with respect to \mbox{$\msf{A}_2\!:\!\msf{B}_2$}) for any operator $\sigma$ that is separable (with respect to \mbox{$\msf{A}_1\!:\!\msf{B}_1$}). Analogously, $\Lambda$ is \emphb{PPT-preserving} if $\Lambda(\sigma)$ is PPT for any PPT operator $\sigma$.
\end{definition}

By definition, the non-entangling channels comprise the largest class of channels which cannot create entanglement, while the PPT-preserving channels are exactly those that cannot generate an NPT state from PPT input states. It is clear that every separable (resp.\ PPT) map is separable-preserving (resp.\ PPT-preserving). We remark that the classes of separable-preserving and PPT-preserving maps are distinct from the classes separable and PPT maps. Indeed, the bipartite swap operator is clearly both separable- and PPT-preserving, but its corresponding Choi representation is proportional to the maximally entangled state which is neither separable nor PPT.

We will also consider the following set of maps that is even more restricted than the non-entangling maps which nonetheless contains the set of LOCC channels.

\begin{definition}
  A completely positive map $\Lambda:\calL(\calH_{\msf{A}_1\msf{B}_1})\rightarrow\calL(\calH_{\msf{A}_2\msf{B}_2})$ is said to be \emphb{dually non-entangling} if both $\Lambda$ and its dual map $\Lambda^*$ are separable-preserving.
\end{definition}

Note that a map is completely positive if and only if its dual map is, but the dual of a trace-preserving map is not necessarily trace-preserving. Dually non-entangling channels form a subset of the non-entangling channels which still contains all of LOCC. Indeed, the dual of any separable map is clearly separable, hence separable maps are dually non-entangling (and thus LOCC channels are as well).

Any quantum channel is a linear map that must not only preserve positivity, but it must also preserve positivity when it acts on the systems tensored with any ancilla space for linear map to represent a physical manipulation of quantum states. Motivated by this requirement of complete positivity for a quantum channel, in the following we define a notion of \emph{complete non-entangling}. The analogous definition is provided for PPT-preserving maps. We note that completely non-entangling maps (resp.\ completely PPT-preserving maps) are exactly those maps that are separable (resp.\ PPT). This will be discussed and proven in Section \ref{sec:kresourcepreserving}.

\begin{definition}
 Let $\Lambda:\calL(\calH_{\msf{A}_1\msf{B}_1})\rightarrow\calL(\calH_{\msf{A}_2\msf{B}_2})$ be a completely positive map and let $k\geq1$ be an integer. Then $\Lambda$ is \emphb{$k$-separable-preserving} (or \emphb{$k$-non-entangling}) if, for all $k$-dimensional systems $\calH_{\msf{A}}=\calH_{\msf{B}}=\CC^k$, the map $\Lambda\otimes\id_{\msf{A}\msf{B}}$ is separable-preserving (with respect to $\msf{A}_1\msf{A}:\msf{B}_1\msf{B}$ and $\msf{A}_2\msf{A}:\msf{B}_2\msf{B}$). If $\Lambda$ is $k$-non-entangling for all integers $k$, then $\Lambda$ is said to be \emphb{completely non-entangling}.
\end{definition}

\begin{definition}
 Let $\Lambda:\calL(\calH_{\msf{A}_1\msf{B}_1})\rightarrow\calL(\calH_{\msf{A}_2\msf{B}_2})$ be a completely positive map and let $k\geq1$ be an integer. Then $\Lambda$ is \emphb{$k$-PPT-preserving} if, for all $k$-dimensional systems $\calH_{\msf{A}}$ and $\calH_{\msf{B}}$, the map $\Lambda\otimes\id_{\msf{A}\msf{B}}$ is PPT-preserving. If $\Lambda$ is $k$-PPT-preserving for all integers $k$, then it is said to be \emphb{completely PPT-preserving}.
\end{definition}


For the final family of channels that we will use, we recall the notion of entanglement witnesses.  A hermitian operator \mbox{$W\in\mc{L}(\mc{H}_{\msf{A}\msf{B}})$} is an entanglement witness with respect to \mbox{$\msf{A}:\msf{B}$} if it holds that $\text{Tr}(W\rho)\geq 0$ for all separable states $\rho\in\mc{S}(\mc{H}_{\msf{A}:\msf{B}})$.  Since the Choi operator for any separable map is separable, we can use entanglement witnesses on the level of Choi operators to detect non-separable maps.  Even stronger, it is known that entanglement witnesses provide a full characterization of separability in the following sense. A positive operator $X$ is separable if and only if $\text{Tr}(WX)\geq 0$ holds for all entanglement witnesses $W$ \cite{Horodecki1996}.

\begin{figure}
\centering
 \includegraphics[width=.6\textwidth]{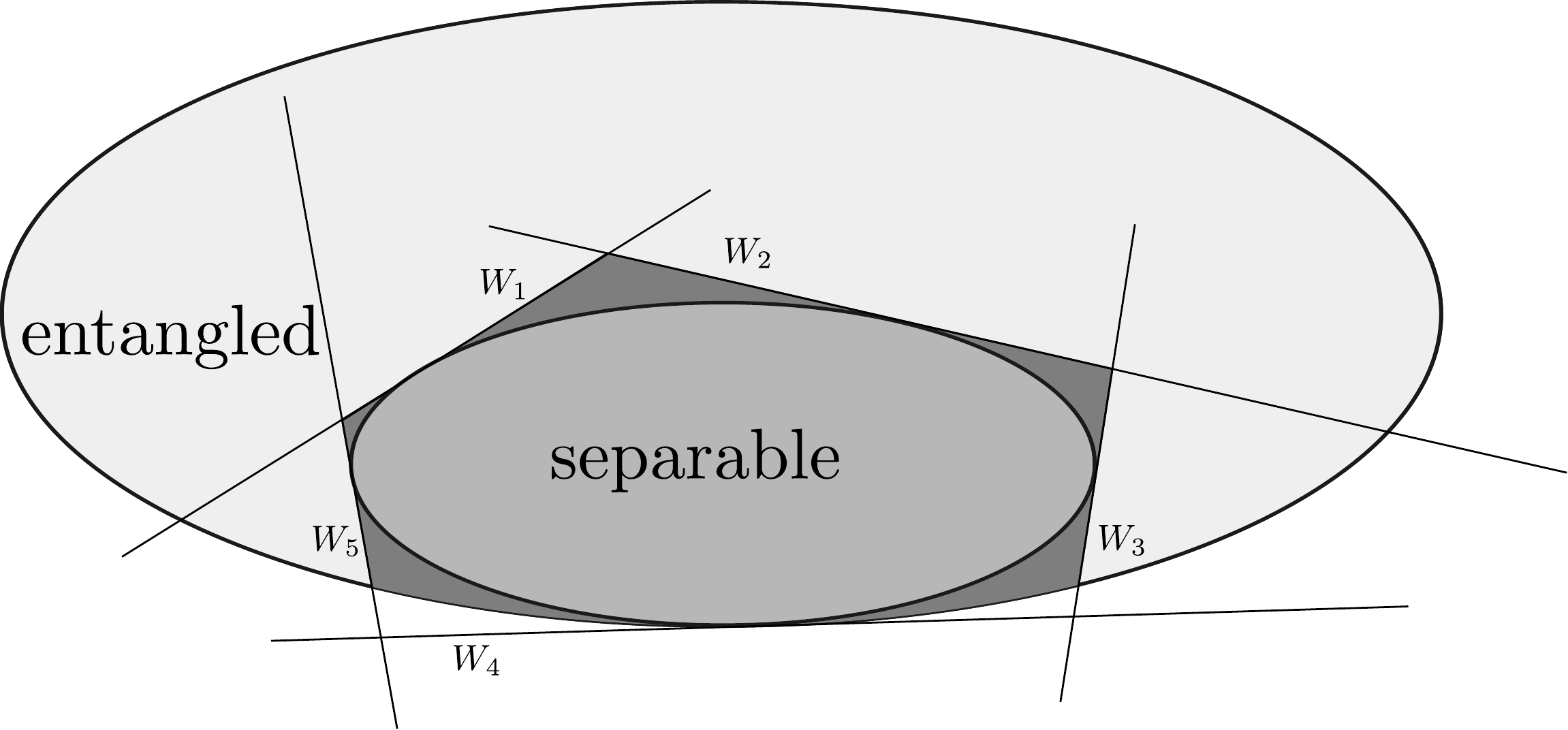}
 \caption{Given a finite number of entanglement witnesses $\{W_i\}_{i=1}^n$ on $\calL(\calH_{\msf{A}_2\msf{A}_1:\msf{B}_2\msf{B}_1})$, the set of entangling undetected channels have Choi representations that are not witnessed by any $W_i$. This polytope clearly contains all separable channels.}
 \label{fig:separablepolytope}
\end{figure}

Geometrically, every entanglement witness can be seen as a hyperplane in the space of hermitian operators that separates the positive convex cone of separable operators (see Fig.\ \ref{fig:separablepolytope}).  Furthermore, the set of separable operators is not a polytope \cite{Ioannou2006, Guhne2007}, and therefore, any finite set of entanglement witnesses is insufficient to determine separability.  Therefore, one obtains a strict relaxation of the separability constraint by requiring that $\text{Tr}(WX)\geq 0$ hold for not all witnesses, but instead for just some finite subset of witnesses.  This idea motivates the following type of operations.
\begin{definition}
Let $\{W_i\}_{i=1}^n\subset\calL(\calH_{\msf{A}_2\msf{B_2}\msf{A_1}\msf{B}_1})$ be a finite collection of entanglement witnesses with respect to \mbox{$\msf{A}_2\msf{A}_1\!:\!\msf{B}_2\msf{B}_1$}.  A completely positive map $\Lambda:\calL(\calH_{\msf{A}_1\msf{B}_1})\rightarrow\calL(\calH_{\msf{A}_2\msf{B}_2})$ is said to be \emphb{entangling undetected} by $\{W_i\}_{i=1}^n$ if it holds that $\text{Tr}(W_i J(\Lambda))\geq 0$ for all $i\in\{1,\cdots,n\}$.
\end{definition}
\noindent Note that for any convex set in $\mbb{R}^d$, a containing polytope can be constructed that approximates the set to arbitrary precision \cite{Bronstein2008}.  Thus for any $\epsilon>0$, there exists a collection of entanglement witnesses $\{W_i\}_{i=1}^n$ such that a non-separable $\Lambda$ is entangling undetected by $\{W_i\}_{i=1}^n$ only if $\lVert J(\Lambda)-\Omega\rVert_1<\epsilon$ for all operators $\Omega\in\mc{L}(\mc{H}_{\msf{A}_2\msf{B}_1\msf{A}_1\msf{B}_1})$ that are separable with respect to $\msf{A}_2\msf{A}_1:\msf{B}_2\msf{B}_1$, where $\lVert\cdot\rVert_1$ denotes the trace norm.


\section{Entanglement measures under maps beyond LOCC}
\label{sec:nonenttransformations}

In this section we explore what kinds of transformations among states are possible under non-entangling channels but not possible under LOCC. This is done by finding transformations under non-entangling channels that increase some entanglement measure. We examine which entanglement measures retain their monotonicity under non-entangling channels, dually non-entangling channels, PPT-preserving channels.

\subsection{Entanglement measures under non-entangling maps}

We first explore certain entanglement measures that can be increased with non-entangling channels, and show how certain other measures can yield conditions for conversion under such maps. Consider the relative entropy of entanglement \cite{Vedral1997,Vedral1998},
\[
 E_R (\rho): = \min_{\sigma\in\calS} S(\rho\Vert\sigma),
\]
where $S(\rho\Vert\sigma):=\Tr(\rho\log\rho)-\Tr(\rho\log\sigma)$ is the relative entropy of $\rho$ and $\sigma$. This is an entanglement measure that is clearly also monotonic under non-entangling channels, a fact that follows directly from the joint monotonicity of the relative entropy, i.e.\ $S(\Lambda(\rho)\Vert\Lambda(\sigma))\leq S(\rho\Vert\sigma)$ for every CPTP map $\Lambda$ \cite{Lindblad1975}. Another entanglement measure that retains its monotonicity under non-entangling channels is the \emph{robustness of entanglement} \cite{Vidal1998}, which is defined as
\begin{equation}
 R(\rho):=\min_{\sigma\in\calS} R(\rho\Vert\sigma),
\end{equation}
where, for any separable state $\sigma\in\calS$,
\begin{equation}
 R(\rho\Vert\sigma) := \min\{s \, :\, \rho + s\sigma \text{ is separable}\}.
\end{equation}
Indeed, if $\Lambda$ is a non-entangling channel, it is clearly seen that $R(\Lambda(\rho)\Vert\Lambda(\sigma))\leq R(\rho\Vert\sigma)$ holds for any separable $\sigma$. For pure states, the robustness reduces to \cite{Vidal1998}
\begin{equation}\label{eq:Rpsi}
 R(\psi) = \Bigl(\sum_i\sqrt{\lambda_i}\Bigr)^2 -1.
\end{equation}
In particular, it holds that $R(\phi^+_k)=k-1$ for any integer $k\geq2$.

Here we show that the robustness yields necessary and sufficient conditions for certain transformations among pure states under non-entangling channels. Whereas pure state transformation under LOCC is governed strictly by majorization of Schmidt coefficients \cite{Nielsen1999}, certain transformations among pure states under non-entangling channels are possible if and only if the robustness of one is greater than the robustness of the other. As Theorem \ref{thm:phik+topsi} shows, this holds when the input state is a maximally entangled pure state.

\begin{theorem}\label{thm:phik+topsi}
 Let $\psi$ be a pure state. There exists a non-entangling channel $\Lambda$ such that $\Lambda(\phi_k^+)=\psi$ if and only if $R(\psi)\leq R(\phi_k^+)$.
\end{theorem}

\begin{proof}
Since we have already pointed out that $R$ cannot increase under non-entangling channels, it remains to explicitly construct a channel achieving the transformation whenever $R(\psi)\leq R(\phi_k^+)$. This is done by the channel defined as
 \begin{equation}
  \Lambda(X) = \Tr(\phi_k^+X)\psi + \Tr((\I-\phi_k^+)X)\sigma,
 \end{equation}
where $\sigma$ is the separable state for which $R(\psi) = R(\psi\Vert\sigma)$. To see that this channel is non-entangling, note that $\Tr(\phi_k^+\rho)\leq 1/k$ holds for all $\rho\in\mathcal{S}$. Therefore, for all $\rho\in\mathcal{S}$, it holds that $\Lambda(\rho)=p\psi + (1-p)\sigma$ for some $p\leq 1/k$. Note that $1/k=1/(1+R(\phi^+_k))$ and that $p\psi + (1-p)\sigma\in\mathcal{S}$ for any $p\leq1/(1+R(\psi))$ by construction. Hence the above channel is non-entangling whenever $R(\psi)\leq R(\phi_k^+)$.
\end{proof}

Theorem \ref{thm:phik+topsi} indicates that the Schmidt rank of a pure state (which cannot be increased under LOCC or separable channels \cite{Terhal2000a}) can in certain cases be increased under non-entangling channels. In fact, for any integers $k$ and $d$ with $k<d$, there exists a pure state with Schmidt rank $k$ that can be transformed to a state with Schmidt rank $d$. This fact is formalized in the following corollary, which follows directly from Theorem \ref{thm:phik+topsi}.

\begin{corollary}
 For all integers $k$ and $d$ with $2\leq k<d$, there exists a non-entangling channel $\Lambda$ such that $\Lambda(\phi^+_k)$ is a pure state with Schmidt rank $d$.
 \end{corollary}

The above result demonstrates that the Schmidt rank is not monotonic under non-entangling channels. It is natural to ask then which other entanglement measures remain monotonic under non-entangling channels and which ones do not. We now investigate other entanglement measures that can be increased under non-entangling channels. The R\'enyi entropy of entanglement of order~$\alpha$ is defined in terms of the Schmidt coefficients of pure states as
\begin{equation*}
 E_{\alpha}(\psi) = H_{\alpha}(\lambda)
\end{equation*}
for any $\alpha\in[0,+\infty]$, where $H_{\alpha}$ are the R\'enyi entropies
\[
 H_{\alpha}(\lambda) = \left\{\begin{array}{ll}
                               \frac{1}{1-\alpha}\log\sum_{i}\lambda_i^\alpha& \alpha\neq0,1,+\infty\\
                               -\sum_{i}\lambda_i\log\lambda_i & \alpha =1\\
                               \log\operatorname{rank}(\lambda) & \alpha=0\\
                               -\log\max_i\{\lambda_i\}& \alpha=+\infty.
                              \end{array}
\right.
\]
These quantities are Schur concave as functions of the Schmidt coefficients. Hence the function~$E_{\alpha}$ is an entanglement measure for pure states for any $\alpha$.  Notice that the Schmidt rank of a pure state is in direct correspondence with the R\'enyi entropy of order zero. Since we have seen that~$E_0$ can be increased under non-entangling channels, it is natural to consider whether $E_{\alpha}$ is not monotonic under non-entangling channels for other values of $\alpha$ as well. The R\'enyi entropy of order~1 reduces to the standard entropy of entanglement,
\[
 E_1(\psi) = E(\psi) = -\sum_{i}\lambda_i\log(\lambda_i),
\]
which is known to be monotonic under non-entangling channels since it coincides with the relative entropy of entanglement $E_R$ for pure states \cite{Vedral1998}. Furthermore, note from \eqref{eq:Rpsi} that the R\'enyi entropy of entanglement of order $1/2$ is a monotonic function of the robustness, since
\[
 E_{1/2}(\psi) = \log(R(\psi)+1),
\]
so $E_{1/2}$ is also monotonic under non-entangling channels by monotonicity of~$R$.

In the following, we show that $E_{\alpha}$ can be increased under non-entangling channels for any $\alpha\in[0,1/2)$, while it cannot be increased under non-entangling channels for any $\alpha\geq1/2$. The proof of the former is actually a simple corollary of Theorem \ref{thm:phik+topsi}. To show this, it suffices to find a pure state $\psi$ with $E_{1/2}(\psi)=E_{1/2}(\phi^+_2)$ but that $E_{\alpha}(\psi)>E_{\alpha}(\phi^+_2)$ for $\alpha<1/2$. Since $E_{1/2}(\phi^+_2)=1$, Theorem \ref{thm:phik+topsi} indicates that the transformation $\phi^+_2\rightarrow\psi$ is possible under non-entangling channels as long as $ E_{1/2}(\psi)\leq 1$.

\begin{corollary}\label{cor:Ealphaincreased}
 The measures $E_{\alpha}$ can be increased by non-entangling channels for any $\alpha\in[0,1/2)$.
\end{corollary}

\begin{proof}
Let $\psi$ be a pure state with Schmidt rank greater than 2 such that $E_{1/2}(\psi)=1$. States with this property certainly exist, since we may take for instance the state with Schmidt coefficients $\lambda=(1-2\epsilon, \epsilon,\epsilon)$ and $\epsilon=1/18$. Note that $E_{\alpha}(\psi)$ is a continuous and non-increasing function of $\alpha$, and that
 \[
 \left.\frac{dE_{\alpha}(\psi)}{d\alpha}\right\rvert_{\alpha=1/2}<0
 \]
for this choice of $\psi$. It follows that $E_{\alpha}(\psi)>1$ for all $\alpha<1/2$. Note however that $E_{\alpha}(\phi^+_2)=1$ holds for all $\alpha$. In particular, $E_{1/2}(\phi^+_2)=1$ so the transformation $\phi^+_2\rightarrow\psi$ is possible with non-entangling channels by Theorem \ref{thm:SchmRnkWithDuallyNonentangling}. This completes the proof.
\end{proof}

To prove that $E_{\alpha}$ is monotonic under non-entangling channels for all $\alpha\geq1/2$, we will show that the value of $E_{\alpha}$ coincides with that of another entanglement measure that is known to be monotonic under non-entangling channels. For every $\alpha\in[0,+\infty)$ and density operators $\rho$ and $\sigma$, the \emph{R\'enyi $\alpha$-relative entropy} of $\rho$ with respect to $\sigma$ is defined as
\[
 S_{\alpha}(\rho\Vert\sigma) = \left\{ \begin{array}{ll}
                                         \frac{1}{\alpha-1}\log\Tr(\rho^{\alpha}\sigma^{1-\alpha}) & \mathop{\mathrm{supp}}\rho\subseteq\mathop{\mathrm{supp}}\sigma \text{ or }\alpha\in[0,1)\\
                                         \Tr(\rho(\log\rho-\log\sigma)) & \mathop{\mathrm{supp}}\rho\subseteq\mathop{\mathrm{supp}}\sigma \text{ and }\alpha=1\\
                                         +\infty & \text{otherwise},
                                       \end{array}
 \right.
\]
where $\mathop{\mathrm{supp}}$ denotes the support of an operator. Note that $S_1=S$ is the standard relative entropy. The \emph{R\'enyi $\alpha$-relative entropy of entanglement} of a state $\rho$ is
\[
 E_{R,\alpha}(\rho):=\inf_{\sigma\in\mathcal{S}}S_{\alpha}(\rho\Vert\sigma),
\]
where the minimization is taken over all separable states $\sigma$. The R\'enyi $\alpha$-relative entropies are known to be \emph{contractive} under CPTP maps for all $\alpha\in[0,2]$, i.e.,
\[
 S_{\alpha}(\Lambda(\rho)\Vert\Lambda(\sigma))\leq S_{\alpha}(\rho\Vert\sigma)
\]
for all CPTP maps $\Lambda$ \cite{Petz1986a}. It readily follows that $E_{R,\alpha}$ is monotonic under non-entangling channels for $\alpha\in[0,2]$. In the case $\alpha=1$, this reduces to the well-known relative entropy of entanglement $E_R$. It is known that $E_{1}(\psi)=E(\psi)=E_{R}(\psi)$ for pure states \cite{Vedral1998}, which establishes that $E_1$ is monotonic under non-entangling channels. Here, we show that $E_{R,\alpha}(\psi) = E_{1/\alpha}(\psi)$ for all $\alpha\in[0,2]$.

\begin{theorem}\label{thm:EalphaisERalphainverse}
 For pure states $\psi$, it holds that $E_{R,\alpha}(\psi) = E_{1/\alpha}(\psi)$ for all $\alpha\in[0,2]$.
\end{theorem}

The proof of Theorem \ref{thm:EalphaisERalphainverse}, which can be found in Section \ref{sec:Salphaproof}, uses similar methods to those in Ref.~\cite{Vedral1998} where it is proven that $E(\psi)=E_{R}(\psi)$. The fact that $E_{\alpha}$ on pure states is monotonic under non-entangling channels for any $\alpha\geq 1/2$ now follows directly from Theorem \ref{thm:EalphaisERalphainverse}.

\begin{corollary}
 For all $\alpha\in[1/2,+\infty]$, the measures $E_{\alpha}$ are monotonic for pure states under non-entangling channels.
\end{corollary}

While majorization provides necessary and sufficient conditions for conversions of pure states under LOCC, whether one can find similar necessary and sufficient conditions for convertibility of arbitrary pure states under non-entangling channels remain unknown. In the following, we provide a sufficient condition for convertibility under non-entangling channels of arbitrary pure states that is independent of the majorization criterion. We first provide the following Lemma (which was first proved in Ref.\ \cite{Shimony1995}).

\begin{lemma}\label{lem:Trpsisigmalambda1}
 For every pure state $\psi$, it holds that $\max_{\sigma\in\mathcal{S}}\Tr(\psi\sigma)=\lambda_1$ where $\lambda_1$ is the largest Schmidt coefficient of $\psi$.
\end{lemma}
\begin{proof}
We may suppose without loss of generality that $\psi$ is in Schmidt form $\ket{\psi}=\sum_{i=1}^n\sqrt{\lambda_i}\ket{ii}$.  We first show that $\Tr(\psi\sigma)\leq\lambda_1$ holds for all $\sigma\in\mathcal{S}$. It will suffice to show that $\Tr(\psi\phi)\leq\lambda_1$ for all separable pure states $\ket{\phi}=\sum_{i,j}u_iv_j\ket{ij}$ with $\sum_{i}\lvert u_i\rvert^2 = \sum_{j}\lvert v_j\rvert^2=1$. Note that $\Tr(\psi\phi) = \lvert\braket{\psi|\phi}\rvert^2$ and
 \begin{align*}
  \lvert\braket{\psi|\phi}\rvert^2 & = \Bigl\lvert\sum_{i}\sqrt{\lambda_i}u_i v_i\Bigr\rvert^2\\
                                 & \leq \lambda_1\Bigl(\sum_{i}\lvert u_i\rvert\lvert v_i\rvert\Bigr)^2 \\
                                 &\leq \lambda_1,
 \end{align*}
 as desired.
 Since $\ket{11}$ is separable and $\lvert\braket{\psi|11}\rvert^2=\lambda_1$, the maximum is achieved. 
\end{proof}

\begin{theorem}\label{thm:Lambdapsi=phi}
 Let $\psi$ and $\phi$ be pure states with Schmidt coefficients $\lambda$ and $\mu$ respectively. If it holds that $1+R(\phi)\leq 1/\lambda_1$, then there exists a non-entangling channel $\Lambda$ such that $\Lambda(\psi)=\phi$.
\end{theorem}
\begin{proof}
 The proof is very similar to that of Theorem \ref{thm:phik+topsi}. The desired channel is given by
 \begin{equation}
  \Lambda(X) = \Tr(\psi X)\phi + \Tr((\I-\psi)X)\sigma
 \end{equation}
where $\sigma\in\mathcal{S}$ is a separable state satisfying $R(\phi)=R(\phi\Vert\sigma)$. It is evident that $\Lambda(\psi)=\phi$, so it remains to show that $\Lambda$ is non-entangling. By construction, it holds that $p\phi + (1-p)\sigma\in\mathcal{S}$ whenever $p\leq 1/(1+R(\phi))$. By Lemma \ref{lem:Trpsisigmalambda1}, it holds that $\Tr(\psi\rho)\leq \lambda_1$ for any separable state $\rho$. If the condition that $1+R(\phi)\leq 1/\lambda_1$ is met, it follows that $\Lambda(\rho)\in\mathcal{S}$ whenever $\rho$ is separable and thus $\Lambda$ is non-entangling.
\end{proof}

\subsection{Entanglement measures under dually non-entangling maps}
\label{sec:duallynonentangling}

We now explore what transformations are possible under the more restrictive class of dually non-entangling channels. Recall that these are the channels $\Lambda$ such that both it and its dual $\Lambda^*$ do not generate entanglement (although $\Lambda^*$ is not necessarily a channel). In particular, we show that the Schmidt rank can be increased by this more restrictive class of transformations.

We first recall some conditions for certain operators to be separable.
For a pure state $\psi$ with Schmidt coefficients $\lambda=(\lambda_1,\dots,\lambda_d)$ (in decreasing order), it holds that $R(\psi\Vert\frac{1}{d^2}\I_{d^2})=d^2\sqrt{\lambda_1\lambda_2}$ (see Appendix B of \cite{Vidal1998}). Hence operators of the form $\psi + s\I_{d^2}$ for $s\geq0$ are separable if and only if $s\geq \sqrt{\lambda_1\lambda_2}$. Additionally, for any hermitian operator $\Delta$ on a bipartite space with local dimensions $d_1$ and $d_2$, the operator $\I_{d_1d_2} + \Delta$ is always separable whenever $\Vert\Delta\rVert\leq 1$, where $\lVert\cdot\rVert$ is the Frobenius norm (see Theorem 1 of Ref.\ \cite{Gurvits2002}). We now state a useful condition for certain types of channels to be dually non-entangling.

\begin{lemma} \label{lem:psiphiduallynonent}
 Let $\psi$ and $\phi$ be pure states with Schmidt coefficients (in decreasing order) $\lambda$ and $\mu$ respectively, and consider the CPTP map $\Lambda:\calL(\CC^k\otimes\CC^k)\to\calL(\CC^d\otimes\CC^d)$ defined by
 \begin{equation}\label{eq:lambdapsiphi}
  \Lambda(X) = \Tr(\psi X)\phi + \Tr((\I_{k^2}-\psi)X)\frac{1}{d^2}\I_{d^2}.
 \end{equation}
If it holds that
 \begin{equation}\label{eq:lambda1nonentangling}
  \lambda_1\leq \frac{1}{d^2\sqrt{\mu_1\mu_2}}
 \end{equation}
 then $\Lambda$ is non-entangling. If it furthermore holds that
 \begin{equation}\label{eq:mu1dualnonentangling}
  d^2\mu_1\leq 1 + \frac{1}{\sqrt{\lambda_1\lambda_2}}
 \end{equation}
 then $\Lambda$ is also dually non-entangling.
\end{lemma}
\begin{proof} Outputs of the channel in \eqref{eq:lambdapsiphi} are of the form $\Lambda(\rho)= p\phi + (1-p)\I_{d^2}/d^2$ for any state $\rho$, where $p=\Tr(\psi\rho)$. From the observation above, this is separable if and only if $p\leq 1/(1+d^2\sqrt{\mu_1\mu_2})$ since $\lVert \phi\rVert =1$. Since $\Tr(\psi\rho)\leq \lambda_1$ for any separable state $\rho$ (see Lemma \ref{lem:Trpsisigmalambda1}), we obtain the desired result that $\Lambda$ is non-entangling whenever the condition in \eqref{eq:lambda1nonentangling} holds.  On the other hand, for any density operator $\rho$, its output under the dual map of $\Lambda$ is given by
\[
 \Lambda^*(\rho) = \left(\Tr(\phi\rho)-\frac{1}{d^2}\right)\psi + \frac{1}{d^2}\I_{k^2}.
\]
Let $\rho$ be a separable state. If it holds that $\Tr(\phi\rho)\leq1/d^2$ then the operator $(d^2\Tr(\phi\rho)-1)\psi+\I_{k^2}$ is separable by Theorem 1 of \cite{Gurvits2002}, since $\lVert\psi\rVert=1$ and $\lvert d^2\Tr(\phi\rho)-1\rvert\leq1$. We may therefore assume that $\Tr(\phi\rho)>1/d^2$, in which case $\Lambda^*(\rho)$ is separable if and only if $d^2\Tr(\phi\rho)\leq 1+1/\sqrt{\lambda_1\lambda_2}$. Since $\Tr(\phi\rho)\leq \mu_1$ holds for any separable state $\rho$, it follows that $\Lambda^*(\rho)$ is separable for all separable states $\rho$, as long as the condition in \eqref{eq:mu1dualnonentangling} holds.
\end{proof}

Now that we have conditions for a channel of the form in \eqref{eq:lambdapsiphi} to be dually non-entangling, we may construct such dually non-entangling channels that increase the Schmidt rank of pure states arbitrarily. The trick is to construct pure states $\psi$ and $\phi$ with Schmidt coefficients $\lambda$ and $\mu$ that satisfy both \eqref{eq:lambda1nonentangling} and \eqref{eq:mu1dualnonentangling}. This is done in the proof of Theorem \ref{thm:SchmRnkWithDuallyNonentangling}.

\begin{theorem}\label{thm:SchmRnkWithDuallyNonentangling}
 For all integers $d,k\geq2$, there exists pure states $\psi$ and $\phi$ with Schmidt ranks equal to $k$ and $d$, respectively, and a dually non-entangling channel $\Lambda$ such that $\Lambda(\psi)=\phi$.
\end{theorem}
\begin{proof}
Let $d,k\geq2$ be integers and choose positive real numbers $\delta$ and $\epsilon$ small enough so that
\begin{align}
 1-\delta           &\leq \frac{1}{1+\frac{d^2}{\sqrt{d-1}}\sqrt{(1-\epsilon)\epsilon}}
 \label{eq:epsdelta1}\\
 \text{and}\qquad
 d^2(1-\epsilon)    &\leq 1+\frac{\sqrt{k-1}}{\sqrt{(1-\delta)\delta}}
 \label{eq:epsdelta2}
\end{align}
are both satisfied. This can be done, since, for example, the values $\delta=d^{-4}$ and $\epsilon=d^{-12}$ satisfy both \eqref{eq:epsdelta1} and \eqref{eq:epsdelta2} (see Appendix \ref{appendix:epsdeltaproof} for proof). Let $\psi$ and $\phi$ be pure states with Schmidt coefficients $\lambda$ and $\mu$ given by
\begin{align}
 \lambda &= \left(1-\delta, \frac{\delta}{k-1}, \dots,\frac{\delta}{k-1}\right)\\
 \text{and}\quad  \mu &= \left(1-\epsilon, \frac{\epsilon}{d-1},\dots,\frac{\epsilon}{d-1}\right)
\end{align}
such that $\psi$ and $\phi$ have Schmidt rank $k$ and $d$, respectively. The desired channel is given by
  \begin{equation}
  \Lambda(X) = \Tr(\psi X)\phi + \Tr((\I_{k^2}-\psi)X)\frac{1}{d^2}\I_{d^2},
 \end{equation}
 which is dually non-entangling by Lemma \ref{lem:psiphiduallynonent} and performs the transformation $\Lambda(\psi)=\phi$.
\end{proof}

There is nothing remarkable about the statement in Theorem \ref{thm:SchmRnkWithDuallyNonentangling} in the case when $k\geq d$. However, if one chooses $2\leq k < d$, Theorem \ref{thm:SchmRnkWithDuallyNonentangling} indicates that it is possible to increase the Schmidt rank of a pure state from $k$ to $d$ by a non-entangling channel. We remark that it does not seem possible to strengthen the above results by considering maps whose outputs are mixtures of $\phi$ and some separable state $\sigma$ other than the identity, as we did in Theorems \ref{thm:phik+topsi} and \ref{thm:Lambdapsi=phi} for the case of non-entangling channels. The reason is that $\sigma$ can be rank-deficient \cite{Gurvits2002} so that there exists some other separable state $\rho\in\mathcal{S}$ such that $\Tr(\sigma\rho) = 0$. With this state, the output of the dual map $\Lambda^*(\rho)$ is entangled and thus $\Lambda$ cannot be dually non-entangling.

Note that a completely positive map $\Lambda$ is separable if and only if its dual map $\Lambda^*$ is separable. As we have seen here, the same is not necessarily true of separability preserving maps, a fact that lead us to explore the class of maps that are dually non-entangling. Similarly, a map $\Lambda$ is PPT if and only if its dual map is PPT. Indeed, for any map $\Lambda$ it is clear that $(\Lambda^*)^{\Gamma}=(\Lambda^{\Gamma})^*$. Since a map is completely positive if and only if its dual is, it follows that $\Lambda^{\Gamma}$ is completely positive if and only if $(\Lambda^*)^{\Gamma}$ is. The condition that a map is PPT-preserving, however, is not equivalent to the condition that its dual map be PPT-preserving. The class of ``dually PPT-preserving'' maps may also be studied in principle, but are not investigated in this paper.

\subsection{Negativity under PPT-preserving channels}
\label{sec:PPTpreservingtransformations}

We now investigate the difference between the class of PPT maps (whose Choi representation is PPT) defined by Rains \cite{Rains1999} and the class of PPT-preserving maps defined in this paper. Despite some confusion in the literature \cite{Audenaert2003,Plenio2007}, these classes are not equivalent.\footnote{Note that a comparison of the distillable entanglement to the entanglement cost under PPT channels has been carried out in Ref.\ \cite{Audenaert2003}, whose authors use the terminology ``operations preserving the positivity of partial transpose'' for what we call PPT maps.} (This distinction goes away, however, if we require that a map be PPT-preserving when tensored with the identity on any ancilla space.) We have seen that these two classes of maps are indeed distinct, as, for example, the swap operation is clearly PPT-preserving but its Choi representation is not PPT. To further distinguish between these classes of maps, we present in this section examples of transformations under PPT-preserving channels that is not possible under PPT channels. This will be done by investigating the \emph{negativity}, a quantity defined as
\[
 N(\rho) = \frac{\Tr\lvert\rho^\Gamma\rvert -1}{2}
\]
for bipartite states $\rho$. It is well-known that the negativity is monotonic under PPT channels (and thus under LOCC as well), so the negativity is an entanglement measure. Here, however, we show that the negativity can be \emph{increased} by PPT-preserving channels. This indicates that the class of PPT-preserving channels is far more powerful for entanglement manipulation than the class of channels whose Choi representations are PPT.

Before proceeding, we first note that states of the form $a\I_{d^2}/d^2+(1-a)\phi^+_d$ are PPT if and only if $a\geq d/(d+1)$. Indeed, note that
\[
 \left(\frac{a}{d^2}\I_{d^2}+(1-a)\phi^+_d\right)^\Gamma = \frac{a}{d^2}\left(\I_{d^2} + \frac{d(1-a)}{a}F_d\right)
\]
which is positive exactly when $a\geq d/(d+1)$, since the eigenvalues of $F_d$ are $\pm1$. We now show that the negativity can be increased by PPT-preserving channels.

\begin{theorem}
The negativity is not a monotone under PPT-preserving channels.
\end{theorem}
\begin{proof} Let $d>3$ be an integer. We construct a PPT-preserving channel $\Lambda$ satisfying $N(\Lambda(\rho))>N(\rho)$ for some bipartite state~$\rho\in\calD(\CC^d\otimes\CC^d)$. Consider the operator
\begin{align*}
 A&=\frac{1}{d+1}\left(d\I_{d^2}+F_d\right) \\&= \frac{d}{d+1}\left(\I_{d^2}+\phi^+_d\right)^\Gamma
\end{align*}
and the corresponding channel $\Lambda:\calL(\CC^d\otimes\CC^d)\rightarrow\calL(\CC^d\otimes\CC^d)$ defined by
\begin{equation}\label{eq:lambdappt}
 \Lambda(X) = \Tr(AX)\frac{1}{d^2}\I_{d^2} +  \Tr((\I_{d^2}-A)X)\phi^+_d.
\end{equation}
By the observation above, it holds that $\Lambda(\rho)$ is PPT whenever $\Tr(A\rho)\geq d/(d+1)$, which certainly holds for any PPT state $\rho$. We conclude that $\Lambda$ is PPT-preserving. Now consider the state
\[
    \rho=\frac{1}{d(d-1)}(\I_{d^2}-F_d),
\]
which is the most entangled Werner state and has negativity equal to $N(\rho)=1/d$. However
\[
 \Lambda(\rho) = \frac{d-1}{d+1}\frac{1}{d^2}\I_{d^2} + \frac{2}{d+1}\phi^+_d,
\]
which has negativity
\[
N(\Lambda(\rho)) = \frac{d-1}{2d}.
\]
As $d>3$, it follows that $N(\Lambda(\rho))> N(\rho)$ for this state.
\end{proof}

Note that for $\rho$ and $\Lambda$ as in the proof above, we have $N(\Lambda(\rho))/N(\rho)=(d-1)/2$. Therefore, this ratio can be made arbitrarily big by taking a large enough $d$. Moreover, a PPT-preserving channel that is not PPT can be constructed from any NPT state~$\rho$ and a corresponding pure state vector $\ket{\psi}$ corresponding to a negative eigenvalue of $\rho^\Gamma$. Indeed, for such a $\rho$ and $\ket{\psi}$, define the operator
\[
 A = \frac{1}{d+1} (d\I_{d^2} +\psi).
\]
The corresponding channel $\Lambda$ of the form in \eqref{eq:lambdappt} is clearly PPT-preserving since it maps every state to a PPT state, which results from the fact that $\Tr(A\sigma)\geq d/(d+1)$ for every PPT state $\sigma$. Nonetheless, this channel is not PPT, since the operator $J(\Lambda)^\Gamma$ is not positive semidefinite. Indeed, note that $(\I_{d^2}-F_d)\otimes\rho^T\geq0$ but that
\begin{align*}
 \Tr\bigl(J(\Lambda)^\Gamma \bigl((\I_{d^2}-F_d)\otimes\rho^T\bigr)\bigr)
 & = \Tr\bigl(J(\Lambda) \bigl((\I_{d^2}-d\phi^+_d)\otimes(\rho^\Gamma)^T\bigr)\bigr)\\
 & =  \left(1-\frac{1}{d}\right) \Tr(A\rho^\Gamma)+ (1-d)(1-\Tr(A\rho^\Gamma))\\
 & <0
\end{align*}
since $\Tr(A\rho^\Gamma)<d/(d+1)$. 


\section{Distillability beyond LOCC}
\label{sec:distillbeyondLOCC}

In order to explore the possible existence of NPT bound entanglement under LOCC, we address the question of distillable entanglement under classes of channels that are strictly larger than LOCC.  As mentioned in the introduction, it has been shown that all NPT states are distillable by PPT channels \cite{Eggeling2001}, but non-entangling channels are independent from this class of maps.  Reference \cite{Brandao2010} has shown that that all entangled states are not only distillable under asymptotically non-entangling channels but also under non-entangling channels.

A state $\rho$ is said to be \emph{distillable} under some class of channels if there exists a sequence of channels $\{\Lambda_n\}$ in this class such that $\lim_{n\to\infty}\lVert\Lambda_n(\rho^{\otimes n})-\phi^+_2\rVert_1=0$, where $\lVert\cdot\rVert_1$ denotes the trace norm. It has been shown that every entangled two-qubit state is distillable by LOCC channels \cite{Horodecki1997a}.  It is tempting to argue that if some class of operations containing LOCC can convert $\rho$ to a two-qubit entangled state, then $\rho$ is distillable using these operations.  However, this argument only holds if the class of operations is closed under tensor products.  More precisely, one requires that $\Lambda^{\otimes n}$ belongs to the class whenever $\Lambda$ does. Indeed, if $\Lambda$ is in a class of operations that is closed under tensor products and it is the case that $\Lambda(\rho)$ is a two-qubit entangled state, then one can use this class of operations to generate the $n$-copy state $\left(\Lambda(\rho)\right)^{\otimes n}=\Lambda^{\otimes n}(\rho^{\otimes n})$ for any $n$. For large enough $n$, this can be further converted to $\phi_2^+$ with arbitrary precision using the LOCC distillation protocol of Ref.~\cite{Horodecki1997a}.  It is true that $\Lambda^{\otimes n}$ is LOCC or PPT whenever $\Lambda$ is LOCC or PPT.  However, for other classes of mappings this will not be true (see below).  We must therefore distinguish between the ability of a class of operations to transform $\rho$ into an LOCC-distillable state on the single-copy level versus its ability to distill $\rho$ into a maximally entangled state in the asymptotic limit.

Although obtaining an LOCC-distillable state is weaker than obtaining the desired state $\phi_2^+$, we believe studying this question is operationally justified.  LOCC operations are generally regarded as being less resource intensive since they do not require the transmission of quantum information.  Thus our approach can be seen as partitioning the distillation of $\phi_2^+$ into two steps: the LOCC part and the non-LOCC part.  In this way we can minimize the quantum communication resources needed to perform the distillation by completing as much of the task as possible using LOCC.  An analogous rationale is applied in other quantum resource theories such as magic-state quantum computation, in which a quantum computation is divided into two parts: the generation of so-called ``magic states,'' and then the manipulation of these states using much simpler stabilizer operations \cite{Bravyi-2005a}.

It is easy to see that not all non-entangling maps are completely non-entangling. A simple example is given by the flip map $\Lambda(\cdot)=F_d \cdot F_d$ (see also Sec.\ \ref{sec:kresourcepreserving} below). However, this does not forbid in principle that $\Lambda^{\otimes n}$ belongs to the class of non-entangling maps whenever $\Lambda$ does. We start this section by answering this question providing an explicit example of a non-entangling map that becomes entangling when taken in two-copy form.

\begin{theorem}\label{thm:superactivationnonentangling}
Non-entangling maps can be superactivated. That is, there exists non-entangling maps $\Lambda$ such that $\Lambda^{\otimes2}$ is not non-entangling.
\end{theorem}
\begin{proof}
We provide an explicit example. Consider the map $\Lambda:\mc{L}(\mbb{C}^2\otimes\mbb{C}^2)\to\mc{L}(\mbb{C}^2\otimes\mbb{C}^2)$ given by
\begin{equation}
\Lambda(X)=\tr[\phi_2^+ X]\phi_2^++\tr[(\mbb{I}-\phi_2^+)X]\frac{1}{2}(\psi^+_2+\psi_2^-),
\end{equation}
where $\psi^\pm_2=(|01\rangle\pm|10\rangle)/\sqrt{2}$ are elements of the Bell basis. This map is non-entangling since $\tr[\phi_2^+X]\leq\frac{1}{2}$ for all separable $X$, and states of the form $p\psi_2^++(1-p)\frac{1}{2}(\psi^+_2+\psi^-_2)$ are separable whenever $p\leq 1/2$, as an entangled Bell-diagonal state requires one eigenvalue greater than $1/2$ \cite{bennett1996,horodecki1996b}. However, the map $\Lambda^{\otimes 2}$ is entangling. To see this, consider the so-called Smolin state \cite{Smolin2001},
\begin{equation}
\rho^{AA':BB'}=\frac{1}{4}\left(\phi_2^{+,AB}\otimes \phi_2^{+,A'B'}+\phi_2^{-,AB}\otimes \phi_2^{-,A'B'}+\psi_2^{+,AB}\otimes \psi_2^{+,A'B'}+\psi_2^{-,AB}\otimes \phi_2^{-,A'B'}\right).
\end{equation}
Despite its appearance, this state is separable in the cut $AA':BB'$ \cite{Smolin2001}.  The action of $\Lambda^{\otimes 2}$ on this state is
\begin{align}
\Lambda^{\otimes 2}(\rho)&=\frac{1}{4}\phi_2^{+,AB}\otimes \phi_2^{+,A'B'}+\frac{3}{16}(\psi^+_2+\psi_2^-)^{AB}\otimes (\psi^+_2+\psi_2^-)^{A'B'}\notag\\
&=\frac{1}{4}\phi_2^{+,AB}\otimes \phi_2^{+,A'B'}+\frac{3}{16}(\ketbra{01}{01}+\ketbra{10}{10})^{AB}\otimes (\ketbra{01}{01}+\ketbra{10}{10})^{A'B'}.
\end{align}
Taking a partial transpose on system $BB'$ yields
\begin{align}
\Lambda^{\otimes 2}(\rho)^\Gamma=\frac{1}{16}F_2^{AB}\otimes F_2^{A'B'}+\frac{3}{16}(\ketbra{01}{01}+\ketbra{10}{10})^{AB}\otimes (\ketbra{01}{01}+\ketbra{10}{10})^{A'B'},
\end{align}
Note that $F_2$ has an eigenvalue $-1$ with eigenvector $\ket{\psi^-_2}$.  Hence,
\begin{align}
\bra{00}^{AB}\bra{\psi^-_2}^{A'B'}\left(\Lambda^{\otimes 2}(\rho)^\Gamma\right)\ket{00}^{AB}\ket{\psi^-_2}^{A'B'}=-\frac{1}{16},
\end{align}
which implies that $\Lambda^{\otimes 2}(\rho)$ is entangled.
\end{proof}

Given the above result, we now turn to the problem of obtaining LOCC-distillable states on the single-copy level. Our techniques will involve maps whose outputs are always mixtures of some operator with the maximally entangled state of two qubits. Note that the state $p\I_4/4+(1-p)\phi^+_2$ is separable if and only if $p\geq2/3$ \cite{Vidal1999}. We first show how any entangled state can be transformed to an LOCC-distillable state using dually non-entangling channels.

\begin{lemma}\label{lem:LambdaWnonentangling}
 Let $W$ be an entanglement witness $W$ such that $\lVert W\rVert_2\leq 1$ (where $\lVert\cdot\rVert_2$ is the Frobenious norm), and let $A=(W+2\I)/3$. Then the channel $\Lambda$ defined by
 \begin{equation}\label{eq:LambdaWnonentangling}
  \Lambda(X)= \Tr(AX)\frac{1}{4}\I_4 + \Tr((\I-A)X)\phi^+_2
 \end{equation}
is dually non-entangling.
\end{lemma}

\begin{proof}
First note that $\lVert W\rVert_2\leq 1$ implies that $-\I\leq W\leq \I$ and thus $0\leq A\leq \I$, so $\Lambda$ is CPTP. Since $W$ is an entanglement witness, it holds that $\Tr(W\sigma)\geq0$ for all separable states $\sigma$. Hence $\Tr(A\sigma)\geq2/3$ holds for any separable state $\sigma$ and therefore, as discussed above, the state $\Lambda(\sigma)$ is separable whenever $\sigma$ is separable.
Thus $\Lambda$ is non-entangling. To conclude that $\Lambda$ is also dually non-entangling, we show that $\Lambda^*(\sigma)$ is separable for all two-qubit states $\sigma$. Note that
\[
 \Lambda^*(\sigma)\propto \I + \frac{\frac{1}{4}-\Tr(\phi^+_2\sigma)}{\frac{1}{2}+\Tr(\phi^+_2\sigma)}W
\]
and that $0\leq\Tr(\phi^+_2\sigma)$ holds for any two-qubit state $\sigma$. Since
\[
 -\frac{1}{2}\leq  \frac{\frac{1}{4}-\Tr(\phi^+_2\sigma)}{\frac{1}{2}+\Tr(\phi^+_2\sigma)} \leq \frac{1}{2}
\]
and $\lVert W\rVert_2\leq1$, it holds that $\Lambda^*(\sigma)$ is separable by Theorem 1 of Ref.\ \cite{Gurvits2002}.
\end{proof}

\begin{theorem}\label{thm:distillnonentangling}
 All entangled states may be converted by dually non-entangling channels to a state that is distillable under LOCC.
\end{theorem}

\begin{proof}
 Let $\rho$ be an arbitrary entangled state. There exists an entanglement witness $W$ such that $\Tr(W\rho)<0$ but that $\Tr(W\sigma)\geq0$ for all separable states $\sigma\in\mathcal{S}$. We may suppose without loss of generality $\lVert W\rVert_2 \leq1$. Indeed, we may otherwise take the witness $W'=W/\lVert W\rVert_2$. Consider the channel $\Lambda$ in \eqref{eq:LambdaWnonentangling} with $A=(W+2\I)/3$, which is dually non-entangling by Lemma \ref{lem:LambdaWnonentangling} and satisfies $\Tr(\Lambda(\rho)\phi^+_2)>1/2$. It follows from Lemma \ref{lem:Trpsisigmalambda1} that $\Lambda(\rho)$ is an entangled state of two qubits and is therefore distillable under LOCC.
\end{proof}

Lastly, while it is known that all non-PPT entangled states are distillable under PPT channels, we next show that all such states can be converted to a state that is distillable under LOCC by only using channels in the more restrictive class of maps that are both PPT and dually non-entangling.

\begin{theorem}\label{thm:distillnonentangling+ppt}
 All NPT states may be converted by channels that are both dually non-entangling and PPT to a state that is distillable under LOCC.  
\end{theorem}

\begin{proof}
 Let $\rho$ denote an arbitrary $d_1\times d_2$ NPT state and let $\ket{\eta}$ denote a normalized eigenvector of $\rho^{\Gamma}$ corresponding to a negative eigenvalue. Then $W=\ketbra{\eta}{\eta}{}^{\Gamma}$ is an entanglement witness that detects $\rho$, since $\Tr(W\rho)<0$ but $\Tr(W\sigma)\geq 0$ for all separable states $\sigma$. Furthermore, note that the Frobenious norm of this witness is $\lVert W\rVert_2=\lVert\ketbra{\eta}{\eta}\rVert_2=1$. It follows from Lemma \ref{lem:LambdaWnonentangling} that the channel $\Lambda$ in \eqref{eq:LambdaWnonentangling} is dually non-entangling. To show that this channel is also PPT, note that the Choi representation of this map is
 \[
  J(\Lambda) = \frac{1}{4}\I_4\otimes A^T + \phi^+_2\otimes(\I_{d_1d_2}-A^T),
 \]
where $A=(W+2\I_{d_1d_2})/3$. The partial transpose of this operator is
\[
 J(\Lambda)^{\Gamma} = \frac{1}{4}\I_4\otimes\ketbra{\overline\eta}{\overline\eta} + \frac{1}{6}(\I_4 + F_2)\otimes(\I_{d_1d_2} - \ketbra{\overline\eta}{\overline\eta}),
\]
where $F_2 = 2(\phi^+_2)^\Gamma$ is the flip operator. This is clearly positive since $-\I_4\leq F_2$. Lastly, note that $\Lambda(\rho)$ is an entangled two-qubit state, so it is distillable under LOCC.
\end{proof}

Finally, we show that all states (even separable ones) can be converted to an LOCC-distillable state under any class of entangling undetected channels.
\begin{theorem}
\label{thm:finite-witness-distill}
 Let $\{W_i\}_{i=1}^n\subseteq\calL(\calH_{\msf{A}_2\msf{B}_2\msf{A}_2\msf{B}_1})$ be a finite number of entanglement witnesses with respect to \mbox{$\msf{A}_2\msf{A}_1\!:\!\msf{B}_2\msf{B}_1$}. Every state can be converted to an LOCC-distillable state using channels that are entangling undetected by $\{W_i\}_{i=1}^n$.
\end{theorem}

\begin{proof}
Let $\Pi=\sum_{i,j=0}^1\ketbra{ij}{ij}$ be the projector onto an arbitrary two-qubit subspace of $\mc{H}_{\msf{A}_2\msf{B}_2}$.  We first note that the projected partial traces of the witnesses $V_i:=\Pi\left[\Tr_{\msf{A}_1\msf{B}_1}(W_i)\right]\Pi$ are again entanglement witnesses for two-qubit states with support on the span of $\{\ket{ij}\}_{i,j=0}^1\subset \calH_{\msf{A}_2:\msf{B}_2}$. Indeed, let $\sigma$ be a separable state with support on $\text{span}\{\ket{ij}\}_{i,j=0}^1$.  Then $\sigma\otimes\I_{\msf{A}_1\msf{B}_1}$ is separable with respect to $\msf{A}_2\msf{A}_1:\msf{B}_2\msf{B}_1$ and thus
 \[
  \Tr(V_i\sigma) = \Tr(W_i (\Pi\sigma\Pi\otimes\I_{\msf{A}_1\msf{B}_1})) =\Tr(W_i (\sigma\otimes\I_{\msf{A}_1\msf{B}_1}))\geq 0.
 \]
Next consider the set of two-qubit density operators which are not detected by the $V_i$.  This is the intersection of a convex set ($\mc{D}$, the set of all two-qubit density operators) with a polyhedron ($\mc{P}$, the set of all hermitian operators satisfying $\text{Tr}(V_iX)\geq 0$ for all $i=1,\cdots n$).  If $\mc{D}$ is contained entirely in the polyhedron, then there obviously exists an entangled two-qubit state $\rho$ not detected by any of the $V_i$.  If $\mc{D}$ is not contained entirely in the polyhedron, then $\mc{D}\cap\mc{P}$ has a facet \cite{Barvinok2002}.  However, the set of two-qubit separable states does not have a facet \cite{Guhne2007}, and thus an entangled state $\rho$ exists in $\mc{D}\cap\mc{P}$.  We then define the map $\Lambda:\calL(\calH_{\msf{A}_1\msf{B}_1})\rightarrow\calL(\calH_{\msf{A}_2\msf{B}_2})$ by $\Lambda(X) = \Tr(X)\rho$. This map has Choi representation $J(\Lambda)=\rho\otimes\I_{\msf{A}_1\msf{B}_1}$ which clearly satisfies $\Tr(W_iJ(\Lambda))\geq0$ for all $i$. The desired result follows since any two-qubit entangled state $\rho$ is distillable under LOCC channels.
\end{proof}


\section{Stochastic Convertibility of Pure States}

\label{sec:stochastic}

In this section we relax the condition that the quantum operation is trace-preserving.  Let $\Lambda_1$ be an arbitrary non-trace-preserving CP map, which by appropriate scaling, we can assume is trace non-increasing; i.e.\ $\text{Tr}(X)\geq\text{Tr}(\Lambda_1(X))\geq 0$ for all $X\geq 0$.  The rescaled map $\Lambda_1$ can always be ``completed'' to a channel by another CP map $\Lambda_2$ so that $\Lambda(\cdot)=\Lambda_1(\cdot)\otimes\ketbra{1}{1}+\Lambda_2(\cdot)\otimes\ketbra{2}{2}$ is CPTP map.  After performing the channel $\Lambda$ on $\rho$, the classical register $\ketbra{i}{i}$ can then be measured so that the post-measurement state $\Lambda_i(\rho)/\text{Tr}[\Lambda_i(\rho)]$ is obtained with probability $\text{Tr}[\Lambda_i(\rho)]$.  Thus, any non-trace-preserving CP map $\Lambda_1$ can be \textit{stochastically} implemented in this way.

While all the discussion thus far has focused on bipartite entanglement, here we consider multipartite systems.  To fix notation, let $\msf{S}_1\cdots\msf{S}_N$ be $N$ systems with joint state space $\mc{H}_{\msf{S}_1\cdots\msf{S}_N}$.  Convertibility between any two states by stochastic LOCC (SLOCC) provides a natural way to classify entanglement.  Under this partitioning of states, two states $\rho$ and $\hat{\rho}$ are said to be equivalent in terms of their entanglement iff $\rho\to\hat{\rho}$ and $\hat{\rho}\to\rho$ by SLOCC \cite{Dur2000}.  When considering pure states, it is well-known that two states $\ket{\psi}$ and $\ket{\hat{\psi}}$ belong to the same entanglement class iff there exists invertible linear operators $A_i$ such that $\ket{\hat{\psi}}=A_1\otimes A_2\otimes\cdots A_N\ket{\psi}$.  A paradigmatic example of two inequivalent states is the three-qubit GHZ state, $\ket{GHZ}=\sqrt{1/2}(\ket{000}+\ket{111})$, and W state, $\ket{W}=\sqrt{1/3}(\ket{100}+\ket{010}+\ket{001})$, which cannot be converted from one to the other using SLOCC \cite{Dur2000}.  In fact, for systems with more than three parties, or for tripartite systems with dimensions greater than $2\otimes 3\otimes 6$, there are an infinite number of inequivalent SLOCC entanglement classes \cite{Chitambar2010}.

We will now show that the situation is dramatically different if any family of entangling undetected operations is considered.  This is similar in spirit to Theorem \ref{thm:finite-witness-distill}, except our proof is simpler in that it does not rely on the convex structure of the set of separable states.  For a CP map $\Lambda:\mc{L}(\mc{H}_{\msf{S}_1\cdots\msf{S}_N})\to\mc{L}(\mc{H}_{\hat{\msf{S}}_1\cdots\hat{\msf{S}}_N})$, its Choi operator is an $N$-partite positive semi-definite operator
\[J(\Lambda)=\Lambda\otimes \text{id}_{\msf{S}_1\cdots\msf{S}_N}(\phi^+_{\msf{S}_1\cdots\msf{S}_N})\in\mc{L}(\mc{H}_{\hat{\msf{S}}_1\msf{S}_1:\cdots:\hat{\msf{S}}_N\msf{S}_N}),\]
where $\phi^+_{\msf{S}_1\cdots\msf{S}_N}=\phi^+_{\msf{S}_1}\otimes\cdots\otimes\phi^+_{\msf{S}_N}$.  Conversely, for any positive operator $\Omega\in\mc{L}(\mc{H}_{\hat{\msf{S}}_1\msf{S}_1:\cdots:\hat{\msf{S}}_N\msf{S}_N})$, the map $\Lambda_\Omega(X)=\text{Tr}_{\msf{S}_1\cdots\msf{S}_N}\left(\Omega[\I_{\hat{\msf{S}}_1\cdots\hat{\msf{S}}_{N}}\otimes X^T]\right)$ is CP.  Direct calculation shows that $J(\Omega)$ is a fully separable operator whenever $\Lambda$ is a separable map, and conversely $\Lambda_\Omega$ is a separable map whenever $\Omega$ is fully separable.  By fully separable, it is meant that $\Omega$ can be expressed as a convex combination of product operators on $\mc{H}_{\hat{\msf{S}}_1\msf{S}_1:\cdots:\hat{\msf{S}}_N\msf{S}_N}$.  Hence, we have a one-to-one correspondence between separable operators and separable maps.

Like in the bipartite case, we say that a hermitian operator $W$ on $\mc{H}_{\hat{\msf{S}}_1\msf{S}_1:\cdots:\hat{\msf{S}}_N\msf{S}_N}$ is an entanglement witness if $\text{Tr}(W\rho)\geq 0$ for all fully separable $N$-partite states $\rho\in\mc{S}(\mc{H}_{\hat{\msf{S}}_1\msf{S}_1:\cdots:\hat{\msf{S}}_N\msf{S}_N})$.
\begin{theorem}
\label{Thm:Stochastic}
Let $\{W_i\}_{i=1}^n$ be a finite number of entanglement witnesses for operators on $\mc{H}_{\hat{\msf{S}}_1\msf{S}_1:\cdots:\hat{\msf{S}}_N\msf{S}_N}$, and let $\ket{\psi}\in\mc{H}_{\msf{S}_1\cdots\msf{S}_N}$ and $\ket{\hat{\psi}}\in\mc{H}_{\hat{\msf{S}}_1\cdots\hat{\msf{S}}_N}$ be any two states.  Then there exists a CP map $\Lambda$ that is entangling undetected by $\{W_i\}_{i=1}^n$ and such that $\Lambda\left(\ketbra{\psi}{\psi}\right)=\ketbra{\hat{\psi}}{\hat{\psi}}$.
\end{theorem}

\begin{proof}
Let $\ket{\psi^{*\perp}}=\ket{a_1}\cdots\ket{a_N}$ be any product state that is orthogonal to $\ket{\psi^*}$.  Define the hermitian operators
\begin{align}
R_i&=\text{Tr}_{\msf{S}_1\cdots\msf{S}_N}\left(W_i[\I_{\hat{\msf{S}}_1\cdots\hat{\msf{S}}_{N}}\otimes\ketbra{\psi^{*\perp}}{\psi^{*\perp}}]\right),&S_i&=\text{Tr}_{\msf{S}_1\cdots\msf{S}_N}\left(W_i[\I_{\hat{\msf{S}}_1\cdots\hat{\msf{S}}_{N}}\otimes(\I-\ketbra{\psi^{*\perp}}{\psi^{*\perp}})]\right),
\end{align}
which by construction are entanglement witnesses on $\mc{H}_{\hat{\msf{S}}_1\cdots\hat{\msf{S}}_N}$.  There must exist a separable state $\omega$ such that $\text{Tr}(\omega R_i)>0$ for all $i=1,\cdots, n$.  Indeed such a state can be constructed as follows.  For each $R_i$, let $\rho_i$ be a separable state such that $\text{Tr}(R_i\rho_i)>0$; since the set of separable states span the space of hermitian operators, we are assured that such a $\rho_i$ exists.  Then take $\omega=\sum_{i=1}^n\rho_i$, and the fact that each $R_i$ is an entanglement witness guarantees that $\text{Tr}(R_i\omega)>0$ for all $i$.  Define
\begin{align}
a&=\min_{i=1,\cdots,n}\text{Tr}(R_i\omega)>0, &b&=\min_{i=1,\cdots,n}\bra{\hat{\psi}}S_i\ket{\hat{\psi}},
\end{align}
and the positive semi-definite operator
\begin{align}
\Omega=\frac{|b|}{a}\omega\otimes\ketbra{\psi^{*\perp}}{\psi^{*\perp}}+\ketbra{\hat{\psi}}{\hat{\psi}}\otimes (\I-\ketbra{\psi^{*\perp}}{\psi^{*\perp}}).
\end{align}
It can be directly seen that $\text{Tr}(W_i\Omega)\geq 0$ for all $i$, and the CP map given by
\[\Lambda(\rho)=\text{Tr}_{\msf{S}_1\cdots\msf{S}_N}\left(\Omega[\I_{\hat{\msf{S}}_1\cdots\hat{\msf{S}}_{N}}\otimes\rho^T]\right)\]
 satisfies $\Lambda\left(\ketbra{\psi}{\psi}\right)=\ketbra{\hat{\psi}}{\hat{\psi}}$.
\end{proof}

Theorem \ref{Thm:Stochastic} has interesting consequences for the problem of tensor rank calculation.  The tensor rank of an $N$-partite state $\ket{\psi}$ is the minimum number of product states whose linear span contains $\ket{\psi}$; i.e.
\begin{equation}
\label{Eq:tensorDefn1}
\text{Tensor rank}\left(\ket{\psi}\right)=\min\left\{r:\ket{\psi}=\sum_{i=1}^r\ket{\varphi^{(\msf{S}_1)}_i}\otimes\cdots\otimes\ket{\varphi^{(\msf{S}_N)}_i}\right\}.
\end{equation}
It is easy to see that the tensor rank of a state can be equivalently characterized as an SLOCC convertibility problem,
\begin{equation}
\label{Eq:tensorSLOCC1}
\text{Tensor rank}\left(\ket{\psi}\right)=\min\left\{r:\ket{GHZ_r^{(N)}}\to \ket{\psi}\text{ by SLOCC}\right\},
\end{equation}
where $\ket{GHZ_r^{(N)}}=\sqrt{1/r}\sum_{i=1}^r\ket{i\cdots i}_{\msf{S}_1\cdots\msf{S}_N}$ \cite{Chitambar2008}.  Given the one-to-one correspondence between separable maps and separable operators, Eq.\ \eqref{Eq:tensorSLOCC1} can be expressed in terms of entanglement witnesses as
\begin{align}
\label{Eq:tensorSLOCC}
\text{Tensor rank}\left(\ket{\psi}\right)=\min\bigg\{r:\text{Tr}_{\msf{S}_1\cdots\msf{S}_N}&\left(\Omega[\I_{\hat{\msf{S}}_1\cdots\hat{\msf{S}}_{N}}\otimes\Phi^{(N)}_r]\right)=\ketbra{\psi}{\psi},\;\;\Omega\geq 0,\notag\\
& \text{Tr}[\Omega W]\geq 0\;\;\forall W\in\mathfrak{W}_{\hat{\msf{S}}_1\msf{S}_1:\cdots:\hat{\msf{S}}_N\msf{S}_N}\bigg\},
\end{align}
where $\Phi^{(N)}_r=\ketbra{GHZ_r^{(N)}}{GHZ_r^{(N)}}$ and $\mathfrak{W}_{\hat{\msf{S}}_1\msf{S}_1:\cdots:\hat{\msf{S}}_N\msf{S}_N}$ is the collection of all entanglement witnesses on $\mc{H}_{\hat{\msf{S}}_1\msf{S}_1:\cdots,\hat{\msf{S}}_N\msf{S}_N}$.

From \eqref{Eq:tensorDefn1}, the tensor rank can be seen as a multipartite generalization of the Schmidt rank.  However, unlike the Schmidt rank, the tensor rank is in general very difficult to compute; in fact already in tripartite systems the problem is NP-Complete \cite{Hastad1990}.  One way to tackle the tensor rank problem is to start from the characterization \eqref{Eq:tensorSLOCC} and move ``beyond'' SLOCC.  Specifically, instead of considering all entanglement witnesses in the minimization of \eqref{Eq:tensorSLOCC}, one could relax the problem and work with some finite subset.  A lower bound on this relaxed problem would also be a lower bound on the tensor rank of $\ket{\psi}$.  However, Theorem \ref{Thm:Stochastic} implies that such a strategy will fail since the minimum for any finite set of witnesses will always be one.  This finding reflects the general difficulty even in computing non-trivial lower bounds for the tensor rank.


\section{\texorpdfstring{$k$}{k}-resource-non-generating maps}
\label{sec:kresourcepreserving}

We now examine the classes of $k$-non-entangling maps and $k$-PPT-preserving maps. As Proposition \ref{prop:completespepreserving} below indicates, the class of completely non-entangling maps coincides with the class of separable maps. In particular, a map $\Lambda$ is completely non-entangling if and only if it is $d$ non-entangling, where $d$ is the dimension of Alice's and Bob's input spaces.  Since a map is separable if and only if its Choi representation is a separable operator, this parallels the notion that complete positivity of a map is equivalent to positivity of its Choi representation. The same holds for maps that are PPT and completely PPT-preserving.

\begin{proposition}\label{prop:completespepreserving}
 Let $\calH_{\msf{A}_1}$ and $\calH_{\msf{B}_1}$ be $d$-dimensional systems and consider a completely positive map $\Lambda:\calL(\calH_{\msf{A}_1\msf{B}_1})\rightarrow\calL(\calH_{\msf{A}_2\msf{B}_2})$. The following are equivalent:
 \begin{enumerate}
  \item[\textup{(i)}] $\Lambda$ is separable (respectively PPT);
  \item[\textup{(ii)}] $\Lambda$ is $d$-non-entangling (respectively $d$-PPT-preserving);
  \item[\textup{(iii)}] $\Lambda$ is completely non-entangling (respectively completely PPT-preserving).
 \end{enumerate}
\end{proposition}

\begin{proof}
The implications (iii)$\Rightarrow$(ii)$\Rightarrow$(i) follow from the definitions. It remains to prove the implication (i)$\Rightarrow$(iii), which follows from the following observations. Note that the tensor product of separable maps is separable, since if $\Lambda:\calL(\calH_{\msf{A}_1\msf{B}_1})\rightarrow\calL(\calH_{\msf{A}_2\msf{B}_2})$ and $\Lambda':\calL(\calH_{\msf{A}'_1\msf{B}'_1})\rightarrow\calL(\calH_{\msf{A}'_2\msf{B}'_2})$ are both separable, then $\Lambda\otimes\Lambda':\calL(\calH_{\msf{A}_1\msf{A}'_1}\otimes\calH_{\msf{B}_1\msf{B}_1'})\rightarrow\calL(\calH_{\msf{A}_2\msf{A}'_2}\otimes\calH_{\msf{B}_2\msf{B}_2'})$ can be written as
\[
 \sum_{i,j}\Phi_i\otimes\Phi_i'\otimes\Psi_j\otimes\Psi_j'
\]
for some CP maps
\begin{align*}
 \Phi_i&:\calL(\calH_{\msf{A}_1})\rightarrow\calL(\calH_{\msf{A}_2})  &
 \Psi_i&:\calL(\calH_{\msf{B}_1})\rightarrow\calL(\calH_{\msf{B}_2}) \\
 \Phi'_i&:\calL(\calH_{\msf{A}'_1})\rightarrow\calL(\calH_{\msf{A}'_2})  &
 \Psi'_i&:\calL(\calH_{\msf{B}'_1})\rightarrow\calL(\calH_{\msf{B}'_2})
\end{align*}
such that $\Lambda=\sum_i\Phi_i\otimes\Psi_i$ and $\Lambda'=\sum_j\Phi'_j\otimes\Psi'_j$, and thus $\Lambda\otimes\Lambda'$ is separable. Analogously, the tensor product of any PPT maps is PPT, which can be clearly seen from the fact that $(\Lambda\otimes\Lambda')^\Gamma = \Lambda^\Gamma\otimes\Lambda'^\Gamma$. In particular, the identity map $\id_{\msf{A}\msf{B}}=\id_{\msf{A}}\otimes\id_{\msf{B}}$ is separable (resp.\ PPT) for any systems $\msf{A}$ and $\msf{B}$, hence the map $\Lambda\otimes\id_{\msf{A}\msf{B}}$ is separable (resp.\ PPT) whenever $\Lambda$ is separable (resp.\ PPT). The desired result follows from the fact that every separable map is non-entangling (and every PPT map is PPT-preserving).
\end{proof}

Note that $1$-non-entangling is the same as non-entangling. (Similarly, $1$-PPT-preserving is the same as PPT-preserving). While it is evident that non-entangling is a strictly larger set than completely non-entangling channels, it is not evident from the definitions that there must exist maps that are $k$-non-entangling for some $k<d$ that are not also completely non-entangling. In the following, we show that there exist $k$-non-entangling maps which are not $(k+1)$-non-entangling for every $k<d$. This shows that the structure of these maps is as rich as it could be.

We first examine a useful condition for determining when a map on $\calL(\CC^d\otimes\CC^d)$ is $k$-non-entangling. Note that any pure state $\ket{u}\in\CC^d\otimes\CC^k$ can be obtained from $\ket{\phi^+_d}$ by some rank-$k$ operator $X$ such that
\[
 \ket{u} = \I_d\otimes X\ket{\phi^+_d}.
\]
A map $\Lambda$ on $\calL(\CC^d\otimes\CC^d)$ is $k$-non-entangling if and only if
\[
 \Lambda\otimes\id_{k^2}(\ketbra{u}{u}\otimes\ketbra{v}{v})
\]
is separable for all pure states $\ket{u},\ket{v}\in\CC^d\otimes\CC^k$. To show that $\Lambda$ is $k$-non-entangling, it suffices to check only when $X$ is a $k$-dimensional projection. A $k$-dimensional projection is a linear operator $P:\CC^d\rightarrow\CC^k$ such that $PP^\dagger=\I_k$. In particular, for such an operator it holds that $\lVert P\rVert = \sqrt{\text{Tr}(\I_k)}=\sqrt{k}$, where $\lVert\cdot\rVert_2$ is the Frobenius norm. The above observations allow us to state the following characterization of $k$-non-entangling maps in Lemma \ref{lem:knonent}.
\begin{lemma}\label{lem:knonent}
 A map $\Lambda$ on $\calL(\CC^d\otimes\CC^d)$ is $k$-non-entangling if and only if the operator
\begin{equation}\label{eq:IXYJLIXY}
 \bigl(\I_{d^2}\otimes P\otimes Q\bigr) J(\Lambda) \bigl(\I_{d^2}\otimes P\otimes Q\bigr)^\dagger
\end{equation}
is separable for all $k$-dimensional projections $P$ and $Q$.
\end{lemma}

The construction of $k$-non-entangling maps that are not $(k+1)$-non-entangling is based on the well known Werner states. The Werner states are a family of symmetric states on $\CC^d\otimes\CC^d$ that are defined by
\begin{equation}
 \rho_{d}(\beta) = \frac{1}{d^2-(\beta+1)}\left(\I_{d^2}-\frac{\beta+1}{d}F_d\right)
\end{equation}
for $-(d+1)\leq \beta\leq d-1$, where $F_d=d (\phi^+_d)^\Gamma$. These states are entangled for $\beta>0$, and  furthermore they are PPT if and only if they are separable. The following lemma presents a fact about Werner states on $\CC^d\otimes\CC^d$ that will be employed in our construction of $k$-non-entangling maps in Theorem \ref{thm:knonentnotdnonent}.

\begin{lemma}\label{lem:kprojectseparable}
Let $k$ and $d$ be integers with $2\leq k<d$ and let $-(d+1)\leq\beta\leq d-1$.  The operator $(P\otimes Q)\rho_{d}(\beta)(P\otimes Q)^\dagger$ is separable for all $k$-dimensional projections $P$ and $Q$ if and only if it holds that $\beta\leq(d-k)/k$.
\end{lemma}

\begin{proof}
 First assume that $\beta\leq(d-k)/k$ and let $P$ and $Q$ be $k$-dimensional projections. Let $\overline{P}$ denote the matrix whose entries are complex conjugate of those of $P$, and let $\ket{\psi}=\overline{P}\otimes Q\ket{\phi^+_d}/\sqrt{c}$ denote the normalized pure state with Schmidt rank $k$ and normalization constant $c>0$ given by
 \begin{align*}
  c &= \Tr(\overline{P}\otimes Q \phi^+_d P^T\otimes Q^\dagger)\\
    &=\frac{1}{d}\Tr(QP^\dagger PQ^\dagger)\\
    &\leq \frac{k}{d},
 \end{align*}
 where the inequality follows from the fact that $QP^\dagger PQ^\dagger\leq\I_k$ with equality if and only if $P=Q$. Hence $c\leq k/d$ with equality if and only if $P=Q$. Since $\lVert\psi^\Gamma\rVert_2=1$, note that
 \begin{equation}\label{eq:PQrhoPR}
  (P\otimes Q)\rho_{d}(\beta)(P\otimes Q)^\dagger\propto \I_{k^2}-c(\beta+1)\psi^\Gamma
 \end{equation}
is separable as long as $c(\beta+1)\leq1$ (by Theorem 1 of \cite{Gurvits2002}). Then it is clearly separable since $\beta\leq(d-k)/k$ by assumption. On the other hand, if $\beta>(d-k)/k$ we may choose $P=Q$ such that
\begin{equation}\label{eq:PQrhoPRGamma}
  \bigl((P\otimes Q)\rho_{d}(\beta)(P\otimes Q)^\dagger\bigr)^\Gamma \propto \I_{k^2}-\frac{k}{d}(\beta+1)\phi^+_k \not\geq0.
 \end{equation}
 We conclude that $(P\otimes Q)\rho_{d}(\beta)(P\otimes Q)^\dagger$ is NPT and therefore entangled.
\end{proof}

We are now ready to present the construction of $k$-non-entangling maps in Theorem \ref{thm:knonentnotdnonent}. The maps presented here are trivially non-entangling since they map every state to a separable state. Nonetheless, the maps defined in Theorem \ref{thm:knonentnotdnonent} are not completely non-entangling and showcase the hierarchy of $k$-non-entangling maps.

\begin{theorem}\label{thm:knonentnotdnonent}
 For all integers $k$ and $d$ with $2\leq k <d$, there exists a $k$-non-entangling map $\Lambda:\calL(\CC^d\otimes\CC^d)\rightarrow\calL(\CC^2\otimes\CC^2)$ that is not $(k+1)$-non-entangling.
\end{theorem}

\begin{proof}
 Let $k$ and $d$ be integers with $2\leq k <d$, and let $\beta>0$ such that
 \begin{equation}\label{eq:betaknonent}
  \frac{d-(k+1)}{k+1}<\beta\leq \frac{d-k}{k}.
 \end{equation}
 Consider the map $\Lambda_\beta$ defined by
 \begin{equation}\label{eq:lambdabetaknonent}
  \Lambda_\beta (X) = \Tr(\rho_d(\beta)X)\ketbra{00}{00} + \Tr((\I_{d^2}-\rho_d(\beta))X)\ketbra{11}{11}.
 \end{equation}
 Note that this map is trivially non-entangling since every output is separable. The corresponding Choi operator of this map is given by
\begin{equation}
 J(\Lambda_\beta) = \ketbra{00}{00}\otimes\rho_{d}(\beta) + \ketbra{11}{11}\otimes (\I_{d^2}-\rho_d(\beta)).
\end{equation}
From Lemma \ref{lem:knonent}, we have that $\Lambda_\beta$ is $k$-non-entangling but not $(k+1)$-non-entangling iff
\begin{equation}\label{eq:IPQJIPQ}
 \bigl(\I_{4}\otimes (P\otimes Q)\bigr) J(\Lambda_\beta) \bigl(\I_{4}\otimes (P\otimes Q)^\dagger\bigr)
\end{equation}
is separable for all $k$-dimensional projections $P$ and $Q$ on $\CC^d$ but entangled for some $(k+1)$-dimensional projections. Note that the operator $\I_{d^2}-\rho_d(\beta)$ is always separable for any $\beta$ (by \cite{Gurvits2002}), and moreover, the operator in Eq.\ (\ref{eq:IPQJIPQ}) is separable iff so is $(P\otimes Q)\rho_{d}(\beta)(P\otimes Q)^\dagger$. Thus, the use of Lemma \ref{lem:kprojectseparable} together with the assumption in Eq.\ (\ref{eq:betaknonent}) completes the proof. 
\end{proof}

Having seen that $k$-non-entangling maps for $1<k<d$ is a strict superset of separable maps and a strict subset of non-entangling maps, this leads one to consider whether these classes of maps allow for bound entanglement. This question remains unanswered, as does the version as to whether all entangled states can be transformed to an LOCC-distillable state under these classes of maps. This is due to the interesting fact that the non-entangling maps constructed used for this task in Sec.\ \ref{sec:distillbeyondLOCC} are not necessarily $k$-non-entangling for every $k<d$.

\begin{proposition}
The dually non-entangling and PPT maps used in the proof of Theorem \ref{thm:distillnonentangling+ppt} that transform the states $\rho_d(\beta)$ with $(d-3)/3<\beta\leq(d-2)/2$ into an LOCC-distillable state are not 3-non-entangling.
\end{proposition}
\begin{proof}
Lemma \ref{lem:kprojectseparable} tells us that for the values of $\beta$ mentioned above, there exist projections on a $3$-dimensional subspace $P$ and $Q$ such that $(P\otimes Q)\rho_{d}(\beta)(P\otimes Q)^\dagger$ is entangled (and NPT). This immediately implies that there exists a pure state $|\eta\rangle$ with Schmidt rank equal to 3 such that $\Tr(\rho_d(\beta)W)<0$, where $W=|\eta\rangle\langle\eta|^\Gamma$ is an entanglement witness. Thus, as argued in Theorem \ref{thm:distillnonentangling+ppt}, the map $\Lambda$ of Eq.\ (\ref{eq:LambdaWnonentangling}) with $A=(W+2\I)/3$ has the property that $\Lambda(\rho_d(\beta))$ is an entangled 2-qubit state. We are now going to show that this map is not 3-non-entangling. Since for $(d-3)/3<\beta\leq(d-2)/2$ the states $\rho_d(\beta)$ are 1-undistillable \cite{Dur2000a}, the above map cannot be separable and, therefore, its Choi operator
\begin{equation}
J(\Lambda)=\left(\frac{2}{3}\frac{\I_4}{4}+\frac{1}{3}\phi^+_2\right)\otimes\I_{d^2}+\frac{1}{3}\left(\frac{\I_4}{4}-\phi^+_2\right)\otimes|\overline\eta\rangle\langle\overline\eta|^\Gamma,
\end{equation}
even though PPT, must be entangled. Now, by Lemma \ref{lem:knonent} the map $\Lambda$ can be seen not to be 3-non-entangling if $J'(\Lambda)=\bigl(\I_{4}\otimes (P\otimes Q)\bigr) J(\Lambda) \bigl(\I_{4}\otimes (P\otimes Q)^\dagger\bigr)$ is still entangled for some choice of projections on a $3$-dimensional subspace $P$ and $Q$. Choosing $P$ and $Q$ to project on the Schmidt bases of $|\overline\eta\rangle$ we obtain that
\begin{equation}
J'(\Lambda)=\left(\frac{2}{3}\frac{\I_4}{4}+\frac{1}{3}\phi^+_2\right)\otimes\I_{9}+\frac{1}{3}\left(\frac{\I_4}{4}-\phi^+_2\right)\otimes|\overline\eta\rangle\langle\overline\eta|^\Gamma,
\end{equation}
which amounts to
\begin{equation}
J(\Lambda)=J'(\Lambda)+\left(\frac{2}{3}\frac{\I_4}{4}+\frac{1}{3}\phi^+_2\right)\otimes\I_{(d-3)^2}.
\end{equation}
Now, since the operator corresponding to the second term in the right-hand-side of the above equation is separable, if $J'(\Lambda)$ were separable too, this would imply that $J(\Lambda)$ is separable, which is a contradiction. Thus, $J'(\Lambda)$ must be entangled.
\end{proof}

Although rather mildly, this suggests that this class of maps could lead to obstructions to distillation allowing for the existence of (NPT) bound entanglement. Thus, we believe that further investigation of the properties of this class deserves future attention.


\section{Proof of Theorem \ref{thm:EalphaisERalphainverse}}
\label{sec:Salphaproof}

This section presents the proof of Theorem \ref{thm:EalphaisERalphainverse}, which shows the equivalence of the $\alpha$-entropy of entanglement and the $(1/\alpha)$-relative entropy of entanglement for $\alpha\in[0,2]$ on pure states. We first state the background necessary for proving this result.

Recall that the R\'enyi $\alpha$-relative entropies $S_{\alpha}$ are monotonic under CPTP maps for $\alpha\in[0,2]$, and the corresponding $\alpha$-relative entropies of entanglement are defined by
\[
 E_{R,\alpha}(\rho) = \inf_{\sigma\in\calS} S_{\alpha}(\rho\Vert\sigma).
\]
When $\rho=\psi$ is a pure state, the $\alpha$-relative entropy reduces (for $\alpha\neq1$) to
\[
 S_{\alpha}(\psi\Vert\sigma) = \frac{1}{\alpha-1}\log\bra{\psi}\sigma^{1-\alpha}\ket{\psi}.
\]
While $S_{\alpha}$ is jointly convex for all $\alpha\in[0,1]$, it is also known to be convex in the second argument for $\alpha\in[0,2]$ (see \cite{Mosonyi2011}). This allows us to state a useful minimization criterion for the $\alpha$-relative entropies of entanglement \cite{Girard2014}. Indeed, for any state $\rho$ and any separable state $\sigma$, it holds that
\[
  E_{R,\alpha}(\rho) = S_{\alpha}(\rho\Vert\sigma)
\]
if and only if it holds that
\begin{equation}\label{eq:salphaderivative}
 \left.\frac{d}{dt} S_{\alpha}(\rho\Vert(1-t)\sigma + t \sigma')\right|_{t=0} \geq 0
\end{equation}
for all other separable states $\sigma'\in\calS$. Furthermore, we can limit our consideration only to pure separable states $\sigma'=\ketbra{\phi}{\phi}$ where $\phi$ is separable.

To show that $E_{R,\alpha}(\psi) = E_{1/\alpha}(\psi)$ holds for all pure states $\psi$ and all $\alpha\in[0,2]$, we must first find a separable state $\sigma$ such that $S_{\alpha}(\psi\Vert\sigma) = E_{1/\alpha}(\psi)$. We then show that \eqref{eq:salphaderivative} holds for all other separable pure states $\sigma'$. The necessary background for computing the derivatives in \eqref{eq:salphaderivative} is given in Section \ref{sec:frechet}. Finally, the proof of Theorem \ref{thm:EalphaisERalphainverse} is given in Section \ref{sec:Ealphaproof}.

\subsection{Directional derivatives of matrix trace functionals}
\label{sec:frechet}

Any real-valued function $f:(0,\infty)\rightarrow\RR$ can be extended to the set of positive definite Hermitian matrices by means of the spectral theorem. In particular, for any $n\times n$ Hermitian matrix $A$ with spectral decomposition $A = \sum_{i=1}^n \alpha_i \ketbra{u_i}{u_i}$, one defines
\[
 f(A)= \sum_{i=1}^n f(\alpha_i) \ketbra{u_i}{u_i}.
\]
Moreover, for any $n\times n$ positive semidefinite matrix $P$, one may define a function $f_P$ on positive definite $n\times n$ matrices as
\[
 f_P(A) = \Tr(Pf(A))
\]
for all positive definite $n\times n$ matrices $A$. Such functions may be extended to certain non-positive-definite matrices in the following way. If $A$ is a $n\times n$ positive semidefinite matrix, one defines $f_P(A)$ as
\begin{equation}\label{eq:fPA}
 f_P(A) = \sum_{\substack{i=1\\ \lambda_i\neq 0}}^n f(\alpha_i) \bra{u_i}P\ket{u_i},
\end{equation}
whenever $\supp(P)\subseteq\supp(A)$, where the sum is taken only over the indices corresponding to nonzero eigenvalues of $A = \sum_{i=1}^n \lambda_i \ketbra{u_i}{u_i}$. By extending the function $f_P$ to be defined on certain positive semidefinite matrices in this way, it is possible to compute the directional derivatives
\[
 \frac{d}{dt} f_P(A+tB)\Bigl|_{t=0^+}.
\]
A derivation and proof of these directional derivatives is provided in Ref.\ \cite{Girard2018}. We provide the details here without proof for clarity.

For any differentiable function $f:(0,\infty)\rightarrow\RR$, its (\emph{first-order}) \emph{divided differences} are defined as
\[
 f^{[1]}(x,y) = \left\{\begin{array}{ll}
                        f'(x) & \text{if }x=y\\
                        \displaystyle\frac{f(x)-f(y)}{x-y}& \text{if }x\neq y
                       \end{array}
\right.
\]
for all $x,y\in(0,\infty)$. Let $f:(0,\infty)\rightarrow\RR$ be a differentiable function. For every  positive semidefinite $n\times n$ matrix $A$, we will define a linear mapping $\Phi_{f,A}$ on the space of $n\times n$ Hermitian matrices as follows. If $A=\diag(\lambda_1,\dots,\lambda_n)$ is diagonal, we define an matrix $n\times n$ Hermitian matrix $D_{f,A}$ of the corresponding matrix of divided differences (restricted to the strictly positive eigenvalues) as
\begin{equation}\label{eq:DfAij}
 \bigl(D_{f,A}\bigr)_{i,j} = \left\{\begin{array}{ll}
                        f^{[1]}(\alpha_i,\alpha_j) & \text{if }\alpha_i,\alpha_j>0\\
                        0& \text{if }\alpha_i=0 \text{ or }\alpha_j=0,
                       \end{array}\right.
\end{equation}
and define the linear mapping $\Phi_{f,A}$ as
\[
\Phi_{f,A}(B) = D_{f,A}\odot B
\]
for all $n\times n$ Hermitian matrices $B$, where $A\odot B$ denotes the entrywise product of the two matrices with matrix entries $(A\odot B)_{i,j} = A_{i,j}B_{i,j}$ for any $n\times n$ Hermitian matrices $A$ and $B$. Thus $\Phi_{f,A}(B)$ has matrix elements
\[
 \bigl(\Phi_{f,A}(B)\bigr)_{i,j} = \left\{\begin{array}{ll}
  f^{[1]}(\alpha_i,\alpha_j) B_{i,j} & \text{if }\alpha_i,\alpha_j>0\\
  0 & \text{if }\alpha_i=0 \text{ or }\alpha_j=0.
  \end{array}
\right.
\]
If $A$ is not diagonal, there is an $n\times n$ diagonalizing unitary matrix $U$ such that the matrix $UAU^\dagger$ is diagonal, and one defines the linear mapping $\Phi_{f,A}$ as
\begin{equation}\label{eq:PhifAB}
 \Phi_{f,A}(B) = U^\dagger\bigl( D_{f,UAU^\dagger}\odot (UBU^\dagger)\bigr) U = U^\dagger\bigl(\Phi_{f,UAU^\dagger}(UBU^\dagger)\bigr) U
\end{equation}
for all $n\times n$ Hermitian matrices $B$, and this mapping is independent of the choice of diagonalizing unitary $U$.

The main result of Ref.\ \cite[Thm.\ 1]{Girard2018} states that directional derivatives of functions of the form $A\mapsto \Tr(Pf(A))$ can be computed at certain non-positive-definite matrices as follows. Suppose $f:(0,\infty)\rightarrow\RR$ is differentiable and satisfies
\begin{equation}\label{eq:limtft}
 \lim_{t\rightarrow0^+}tf(t)=0.
\end{equation}
Let $P$ be a positive semidefinite $n\times n$ matrix, and let $A$ be another positive semidefinite $n\times n$ matrix satisfying $\supp(P)\subseteq\supp(A)$ such that $f_P(A)$ may be defined as in \eqref{eq:fPA}. Let $B$ be any $n\times n$ Hermitian matrix and suppose there exists a positive value $\varepsilon>0$ such that $A+tB$ is positive semidefinite for all $t\in[0,\varepsilon)$. Then one may compute the directional derivative of $f_P$ at $A$ in the direction $B$ as
 \begin{equation}\label{eq:ddtfpAtB}
  \frac{d}{dt} f_P(A+tB) \Bigr|_{t=0^+} = \Tr(P\,\Phi_{f,A}(B)).
 \end{equation}
In the case when $A$ is positive definite, the mapping $\Phi_{f,A}$ is precisely the Fr\'echet derivative of the function $f$ at $A$ (see, e.g., \cite[Sec.\ X.4]{Bhatia1997} and \cite[Thm.\ 3.25]{Hiai2014})

Note that the function $f(x)=\log(x)$ satisfies the condition in \eqref{eq:limtft}, as well as the functions $f_\alpha(x)=x^{1-\alpha}$ for all $\alpha<2$. However, although the condition is \emph{not} satisfied by the function $f(x)=x^{-1}$, in \cite[Sec.\ 5]{Girard2018}, it is shown that the expression in \eqref{eq:ddtfpAtB} is still a lower bound for the directional derivative. That is, if the matrices $P$, $A$, and $B$ satisfy all of the criteria as given above, then
 \begin{equation}\label{eq:ddtfpAtBgeq}
  \frac{d}{dt} \Tr\left(P(A+tB)^{-1}\right) \Bigr|_{t=0^+} \geq \Tr(P\,\Phi_{f,A}(B)),
 \end{equation}
and this inequality is strict in general if $A$ is not positive definite.

Before providing the proof of Theorem \ref{thm:EalphaisERalphainverse}, we make a few more remarks regarding directional derivatives the linear mappings $\Phi_{f,A}$ given above as to how they will be applied here. Let $f:(0,\infty)\rightarrow\RR$ be a differentiable function satisfying \eqref{eq:limtft}, and let $\rho,\sigma,\sigma'\in\calD(\calH)$ be states on a finite-dimensional Hilbert space $\calH$ satisfying $\supp(\rho)\subseteq\supp(\sigma)$. It holds that
\[
 \sigma + t(\sigma'-\sigma) = (1-t)\sigma + t \sigma'
\]
is positive semidefinite for all $t\in[0,1]$, hence the directional derivative of $f_{\rho}$ at $\sigma$ is computed as
\[
 \frac{d}{dt}f_{\rho}((1-t)\sigma + t \sigma')\Bigr|_{t=0^+} = \Tr(\rho\, \Phi_{f,\sigma}(\sigma'-\sigma)).
\]
Moreover, as the mapping $\Phi_{f,\sigma}$ (as defined in \eqref{eq:PhifAB}) is a linear, it can be seen that
\[
 \Phi_{f,\sigma}(\sigma'-\sigma) = \Phi_{f,\sigma}(\sigma')- \Phi_{f,\sigma}(\sigma) = \Phi_{f,\sigma}(\sigma') - \sigma f'(\sigma),
\]
where $f'(\sigma)$ is defined as in \eqref{eq:fPA}.

\subsection{Proof of Theorem }
\label{sec:Ealphaproof}

We now present the proof that $E_{R,\alpha}(\psi) = E_{1/\alpha}(\psi)$ for all pure states $\psi$ and all $\alpha\in[0,2]$. This result is already known when $\alpha=1$ \cite{Vedral1998}, so here we prove it only for $\alpha\neq1$. We will consider the case when $\alpha=0$ separately.

Here, we will make use of the functions $f_\alpha$ for $\alpha\in[0,1)\cup(1,2]$ that are defined by $f_{\alpha}(x)=x^{1-\alpha}$ and whose divided differences are given by
\begin{equation}
 f_{\alpha}^{[1]}(x,y) = \left\{ \begin{array}{ll}
                         \displaystyle\frac{x^{1-\alpha}-y^{1-\alpha}}{x-y} & x\neq y\\
                         (1-\alpha)x^{-\alpha} & x=y
                        \end{array}\right.
\end{equation}
for all $x,y\in(0,\infty)$. For these functions, it is straightforward to check that
\begin{equation*}
  f_{\alpha}^{[1]}\left(\frac{x}{c},\frac{y}{c}\right) = c^\alpha f_{\alpha}^{[1]}(x,y)
\end{equation*}
holds for any positive real constant $c>0$ and any $x,y>0$. We now prove Theorem \ref{thm:EalphaisERalphainverse}.

 As remarked in Section \ref{sec:frechet}, the function $f_{\alpha}$ satisfies condition \eqref{eq:limtft} for all $\alpha<2$. This allows us to compute the directional derivatives of the $\alpha$-relative entropies $S_\alpha(\rho\lVert\sigma)$. For $\alpha=2$, we may compute a lower bound to the derivative as in \eqref{eq:ddtfpAtBgeq}.

\begin{proof}[Proof \textup{(}of Theorem \ref{thm:EalphaisERalphainverse}\textup{)}]
 Let $\psi=\sum_{i}\sqrt{\lambda_i}\ket{ii}$ be a pure state with Schmidt coefficients $\lambda=(\lambda_1,\dots,\lambda_d)$.
 First suppose that $\alpha\neq 0$ and define the following separable density operator
\begin{equation}
 \sigma = \frac{1}{\lVert\lambda\rVert_{1/\alpha}^{1/\alpha}}\sum_{i}\lambda_i^{1/\alpha} \ketbra{ii}{ii},
\end{equation}
where $\lVert\lambda\rVert_{1/\alpha} = \bigl(\sum_{i}\lambda_i^{1/\alpha}\bigr)^{\alpha}$. Note that the $\alpha$-relative entropy of $\psi$ and $\sigma$ is
\begin{align*}
 S_{\alpha}(\psi\Vert\sigma) &= \frac{1}{\alpha-1}\log\langle \psi\vert\sigma^{1-\alpha}\vert\psi\rangle\\
                             &= \frac{1}{\alpha-1}\log\lVert\lambda\rVert_{1/\alpha}\\
                             &= \frac{1}{1-1/\alpha} \log \lVert\lambda\rVert_{1/\alpha}^{1/\alpha},
\end{align*}
from which it follows that $S_{\alpha}(\psi\Vert\sigma) = E_{1/\alpha}(\psi)$. The matrix $\sigma$ has eigenvalues $\lambda_i^{1/\alpha}/\lVert\lambda\rVert^{1/\alpha}_{1/\alpha}$,  and the corresponding matrix $D_{f_{\alpha},\sigma}$ of divided differences as defined in \eqref{eq:DfAij} has non-zero matrix elements given by
\begin{equation}
 \bra{ii} D_{f_{\alpha},\sigma}\ket{jj}  = f_{\alpha}^{[1]}\left(\frac{\lambda_i^{1/\alpha}}{\lVert\lambda\rVert^{1/\alpha}_{1/\alpha}},\frac{\lambda_j^{1/\alpha}}{\lVert\lambda\rVert^{1/\alpha}_{1/\alpha}}\right)= \lVert \lambda\rVert_{1/\alpha}f_{\alpha}^{[1]}\Bigl(\lambda_i^{1/\alpha},\lambda_j^{1/\alpha}\Bigr)
\end{equation}
such that the linear mapping $\Phi_{f_\alpha,\sigma}$ is defined as
\[
 \Phi_{f_\alpha,\sigma} (B) =  \lVert \lambda\rVert_{1/\alpha}\sum_{i,j}  f_{\alpha}^{[1]}\Bigl(\lambda_i^{1/\alpha},\lambda_j^{1/\alpha}\Bigr)\bra{ii}B\ket{jj}\ketbra{ii}{jj}
\]
for all $B$. We first suppose that $\alpha\in(0,1)\cup(1,2)$ (the case $\alpha=2$ must be considered separately). For any other separable state $\sigma'$, we may use \eqref{eq:ddtfpAtB} to compute the necessary derivative
\begin{align*}
 \frac{d}{dt} \bra{\psi}\bigl((1-t)\sigma+t\sigma'\bigr)^{1-\alpha}\ket{\psi}\Bigr|_{t=0^+} & =  \frac{d}{dt} \Tr\left(\psi\,f_{\alpha}\bigl((1-t)\sigma+t\sigma'\bigr)\right)\Bigr|_{t=0^+}\\
  & = \Tr\bigl(\psi \, \Phi_{f_{\alpha},\sigma}(\sigma'-\sigma)\bigr)\\
  & = \bra{\psi}\Phi_{f_{\alpha},\sigma}(\sigma')\ket{\psi} -(1-\alpha)\lVert\lambda\rVert_{1/\alpha},
\end{align*}
where the final equality follows from the fact that $\bra{\psi}\sigma^{1-\alpha}\ket{\psi}= \lVert\lambda\rVert_{1/\alpha}$, the linearity of the mapping $\Phi_{f_{\alpha},\sigma}$, and the fact that $\Phi_{f_{\alpha},\sigma}(\sigma) = (1-\alpha)\sigma^{1-\alpha}$. Thus
\begin{align*}
 \left.\frac{d}{dt} S_{\alpha}\bigl(\psi\Vert(1-t)\sigma + t\sigma'\bigr)\right|_{t=0}
   & = \frac{1}{\alpha-1}  \frac{d}{dt} \log\Bigl(\bra{\psi}\bigl((1-t)\sigma+t\sigma'\bigr)^{1-\alpha}\ket{\psi}\Bigr|_{t=0^+} \\
   & = \frac{1}{\alpha-1}\frac{1}{\bra{\psi}\sigma^{1-\alpha}\ket{\psi}}\left( \bra{\psi}\Phi_{f_{\alpha},\sigma}(\sigma')\ket{\psi}-(1-\alpha)\lVert\lambda\rVert_{1/\alpha}\right)\\
   & = 1-\frac{1}{1-\alpha}\frac{1}{\lVert\lambda\rVert_{1/\alpha}} \bra{\psi}\Phi_{f_{\alpha},\sigma}(\sigma')\ket{\psi}.
\end{align*}
 When $\sigma'=\ketbra{\phi}{\phi}$ is a separable pure state for some $\ket{\phi}=\sum_{i,j}u_iv_j\ket{ij}$ with $\sum_{i}\lvert u_i\rvert^2=\sum_{j}\lvert v_j\rvert^2=1$, it holds that
\begin{align*}
 \Phi_{f_{\alpha},\sigma}(\sigma') & = \lVert \lambda\rVert_{1/\alpha} \sum_{i,j} u_iv_i\overline{u_j}\overline{v_j} f_{\alpha}^{[1]}\bigl(\lambda_i^{1/\alpha},\lambda_j^{1/\alpha}\bigr) \ketbra{ii}{jj}.
\end{align*}
Note that $f^{[1]}_{\alpha}(x,y)/(1-\alpha)\geq0$ for any $x,y>0$. Hence
\begin{align*}
\left\lvert\frac{1}{1-\alpha}\frac{1}{\lVert\lambda\rVert_{1/\alpha}} \bra{\psi}(\Phi_{f_{\alpha},\sigma}(\sigma')\ket{\psi}\right\rvert
 & \leq \sum_{i,j} \lvert u_i\rvert \lvert v_i\rvert \lvert u_j\rvert\lvert v_j\rvert  \frac{\sqrt{\lambda_i\lambda_j}}{1-\alpha}f_{\alpha}^{[1]}\bigl(\lambda_i^{1/\alpha},\lambda_j^{1/\alpha}\bigr) \\
 & \leq \sum_{i,j} \lvert u_i\rvert \lvert v_i\rvert \lvert u_j\rvert\lvert v_j\rvert \\
 & \leq \sum_{i} \lvert u_i\rvert^2 \sum_{i} \lvert v_i\rvert^2 \\
 & = 1,
\end{align*}
where the second inequality is due to Lemma \ref{lem:technical} (see below) and the final inequality follows from the Cauchy-Schwarz inequality. This completes the proof, since we have found that
\begin{equation}\label{eq:dtsalpha}
 \left.\frac{d}{dt} S_{\alpha}\bigl(\psi\Vert(1-t)\sigma + t\sigma'\bigr)\right|_{t=0^+} \geq0
\end{equation}
for all pure separable states $\sigma'$.

To prove the claim in the case when $\alpha=2$, we note that only a lower bound the directional derivative can be provided for the function $f_2(x)= x^{-1}$. In this case, we have
\[
 \frac{d}{dt} \bra{\psi}\bigl((1-t)\sigma+t\sigma'\bigr)^{-1}\ket{\psi}\Bigr|_{t=0^+} \geq \bra{\psi}\Phi_{f_{2},\sigma}(\sigma')\ket{\psi} +\lVert\lambda\rVert_{1/2},
\]
and using the same arguments as above one finds
\begin{align*}
 \left.\frac{d}{dt} S_{2}\bigl(\psi\Vert(1-t)\sigma + t\sigma'\bigr)\right|_{t=0}
  \geq 1+\frac{1}{\lVert\lambda\rVert_{1/2}} \bra{\psi}\Phi_{f_{2},\sigma}(\sigma')\ket{\psi},
\end{align*}
which is non-negative (by Lemma \ref{lem:technical} and the same arguments as above).

Lastly, to prove the claim when $\alpha=0$, set $\sigma=\ketbra{11}{11}$. Note that $S_{0}(\psi\Vert\sigma) = -\log\lambda_1$, where $\lambda_1$ is the largest Schmidt coefficient of~$\psi$. For any other separable state $\sigma'\in\calS$, it holds that $\bra{\psi}\sigma'\ket{\psi}\leq\lambda_1$ (by Lemma \ref{lem:Trpsisigmalambda1}) and thus
\begin{align*}
 S_0(\psi\Vert\sigma')  &= -\log\bra{\psi}\sigma'\ket{\psi}
                                       \geq -\log\lambda_1 = S_0(\psi\Vert \sigma),
\end{align*}
from which it follows that $E_{R,0}(\psi) = -\log\lambda_1$. Since $E_{+\infty}(\psi) = -\log\lambda_1$, the result follows.
\end{proof}

The proof of Theorem \ref{thm:EalphaisERalphainverse} above in the case when $\alpha\neq0$ relies on the following lemma, which we prove below.

\begin{lemma}\label{lem:technical}
 Let $\alpha\in(0,1)\cup(1,2]$. For all $p,q\in(0,1]$, it holds that
 \begin{equation}
  0\leq \frac{\sqrt{pq}}{1-\alpha}f_{\alpha}^{[1]}(p^{1/\alpha},q^{1/\alpha}) \leq 1,
 \end{equation}
 where $f_\alpha:(0,\infty)\rightarrow\RR$ is the function defined by $f_{\alpha}(x)=x^{1-\alpha}$.
\end{lemma}
The proof is trivial in the case when $\alpha=2$. Indeed, $f_{2}(x)=x^{-1}$ and thus for $p\neq q$ it holds that
\[
 -\sqrt{pq}f_{2}^{[1]}(\sqrt{p},\sqrt{q}) = -\sqrt{pq}\frac{\frac{1}{\sqrt{p}}-\frac{1}{\sqrt{q}}}{\sqrt{p}-\sqrt{q}} =1.
\]
The proof in the case when $\alpha\neq2$ is a bit more technical and requires the following pair of results.
\begin{lemma}\label{lem:xpyp}
 For all $x,y>0$ with $x\neq y$ and $r\in(-1,0)\cup(0,1)$, it holds that
 \begin{equation}
  \frac{1}{r}\frac{x^r-y^r}{x-y}\leq  (\sqrt{xy})^{r-1}.
 \end{equation}
\end{lemma}
\begin{proof}
For $x,y>0$ and $r\in(-1,0)\cup(0,1)$, we have the following integral representations:
 \begin{align}
   rx^{r-1} & =\frac{\pi}{\sin(r\pi)}\int_{0}^\infty \frac{t^r}{(x+t)^2}dt\label{eq:pxp}\\
   \frac{x^r-y^r}{x-y} &= \frac{\pi}{\sin(r\pi)}\int_{0}^\infty \frac{t^r}{(x+t)(y+t)}dt.   \label{eq:xpyp}
 \end{align}
Let $x,y>0$ with $x\neq y$. Note that $x+y\geq 2\sqrt{xy}$ and thus
\begin{align*}
 (x+t)(y+t) &= xy + t(x+y)+ t^2 \\&\geq xy + 2\sqrt{xy}t+ t^2 = (\sqrt{xy}+t)^2
\end{align*}
holds for all $t\geq 0$. Let $r\in(-1,0)\cup(0,1)$ and note that $r\sin(r\pi)\geq 0$. Thus
\begin{align*}
\frac{1}{r}\frac{x^r-y^r}{x-y}
     & = \frac{1}{r}\frac{\pi}{\sin(r\pi)}\int_{0}^\infty \frac{t^r}{(x+t)(y+t)}dt\\
     & \leq \frac{1}{r}\frac{\pi}{\sin(r\pi)}\int_{0}^\infty \frac{t^r}{(\sqrt{xy}+t)^2}dt    \\
     & = (\sqrt{xy})^{r-1},
\end{align*}
where we use the representations in \eqref{eq:pxp} and \eqref{eq:xpyp}.
\end{proof}

\begin{lemma}\label{lem:pneqq}
 For all $p,q\in(0,1]$ with $p\neq q$, and all $\alpha\in(0,1)\cup(1,2)$, it holds that
 \begin{equation}
  \frac{\sqrt{pq}}{1-\alpha}\frac{p^{(1-\alpha)/\alpha}-q^{(1-\alpha)/\alpha}}{p^{1/\alpha}-q^{1/\alpha}} \leq 1.
 \end{equation}
\end{lemma}
\begin{proof}
 Setting $r=1-\alpha$, $x=p^{1/\alpha}$, and $y=q^{1/\alpha}$, an application of of Lemma \ref{lem:xpyp} yields
 \begin{align*}
  \sqrt{pq}\frac{1}{1-\alpha}\frac{p^{(1-\alpha)/\alpha}-q^{(1-\alpha)/\alpha}}{p^{1/\alpha}-q^{1/\alpha}}
     & \leq \sqrt{pq}\Bigl(\sqrt{p^{1/\alpha}q^{1/\alpha}}\Bigr)^{-\alpha}\\
     & = 1 ,
 \end{align*}
as desired.
\end{proof}

A proof of Lemma \ref{lem:technical} in the case $\alpha\neq2$ now follows.
\begin{proof}[Proof \textup{(}of Lemma \ref{lem:technical}\textup{)}] As the statement has been shown to be true for $\alpha=2$, we may assume that $\alpha\neq2$.
 If $p\neq q$, we can apply Lemma \ref{lem:pneqq}. Otherwise we have
 \[
  f_{\alpha}^{[1]}(p^{1/\alpha},p^{1/\alpha}) = (1-\alpha)(p^{1/\alpha})^{-\alpha} = \frac{1-\alpha}{p}
 \]
and thus $\frac{p}{1-\alpha}f_{\alpha}^{[1]}(p^{1/\alpha},p^{1/\alpha}) = 1$, which completes the proof.
\end{proof}


\section{Discussion}
\label{sec:discussion}

Although LOCC maps provide the most physically meaningful choice of free operations for the resource theory of entanglement, from a resource-theoretic point of view, other choices of operations still provide consistent and well-defined theories.  Furthermore, insight into LOCC is gained by studying more general resource theories since impossibility results for the latter imply impossibility results for the former.  We have examined many different classes of non-entangling channels that are still larger than the class of LOCC (and separable) channels.  While many entanglement measures are still monotonic under such non-entangling channels, we have presented examples of entanglement measures whose monotonicity fails to hold.  This allowed us to find transformations among entangled states that are possible under non-entangling maps but not possible under LOCC (or separable maps).  PPT channels have also been widely studied in the context of entanglement theory and, in a similar vein, here we have also shown that, perhaps surprisingly, the negativity is no longer a monotone when this class is extended to the set of PPT-preserving channels.

Although we believe that resource theories of entanglement under different classes of operations is an interesting research topic in itself, one of the motivations behind this present work is to understand more deeply the phenomenon of bound entanglement.  The question of whether NPT bound entanglement exists is one of the most challenging open problems in the field.  If NPT bound entanglement does exist, this could be proven by finding some NPT entangled state that still remains undistillable by a larger class of operations than LOCC.  So far, only PPT operations had been considered in this context \cite{Eggeling2001}. Here, we have systematically studied the extraction of entanglement under hierarchies of operational classes that go beyond LOCC. Although this does not necessarily exclude the phenomenon of bound entanglement, we have proven that in most of these classes every entangled state can be converted to a state that is LOCC-distillable.
The only class for which we have been unable to prove that all NPT states can be converted to an LOCC-distillable state is that of $k$-non-entangling maps. Thus, further investigation into the properties of this class could be a promising route to answer the long-standing question of NPT bound entanglement. In particular, even though we have proved that the class of $k$-non-entangling maps is a strict superset of LOCC and SEP maps, this alone does not necessarily imply that there exist pure-state transformations achievable by the former class of maps which are impossible by the latter. It would be interesting to study if such examples exist and, if so, which entanglement measures loose their monotonicity under $k$-non-entangling maps.
\\

EC is supported by the National Science Foundation (NSF) Early CAREER Award No.\ 1352326. JdV acknowledges support from the Spanish MINECO through grants MTM2017-84098-P and MTM2017-88385-P and from the Comunidad de Madrid through grant QUITEMAD-CM P2018/TCS­4342. MG is acknowledges support from an Izaak Walton Killam Memorial Scholarship and an Alberta Innovates--Technology Futures (AITF) Graduate Student Scholarship. GG acknowledges support from the Natural Sciences and Engineering Research Council of Canada (NSERC).


\bibliographystyle{mybst}
\bibliography{BeyondLOCC}

\appendix


\section{Details for proof of Theorem \ref{thm:SchmRnkWithDuallyNonentangling}}
\label{appendix:epsdeltaproof}

Here we show that the values $\delta=d^{-4}$ and $\epsilon=d^{-12}$ satisfy both \eqref{eq:epsdelta1} and \eqref{eq:epsdelta2}. Indeed, for these values of $\epsilon$ and $\delta$, it is straightforward to check that
\[
\frac{d^2}{\sqrt{d-1}}\sqrt{(1-\epsilon)\epsilon} \leq d^{-4}
\]
and thus
\begin{align*}
 (1-\delta)\left(1+\frac{d^2}{\sqrt{d-1}}\sqrt{(1-\epsilon)\epsilon} \right)
 &\leq  (1-d^{-4})(1+d^{-4})\\
 &\leq 1,
\end{align*}
which implies that \eqref{eq:epsdelta1} holds. On the other hand, it clear that
\begin{align*}
 d^2(1-\epsilon)\sqrt{(1-\delta)\delta}
 &\leq d^2\sqrt{\delta} \\
 &= 1.
\end{align*}
Furthermore, it holds that $1\leq \sqrt{k-1}$ since $k\geq2$. It follows that
\begin{align*}
 d^2(1-\epsilon) \leq  \frac{\sqrt{k-1}}{\sqrt{(1-\delta)\delta}}
\end{align*}
%
%
which proves that \eqref{eq:epsdelta2} holds as desired.

\end{document}